\documentclass[10pt]{article}      
 \usepackage[top=0.8in, bottom=0.7in, left=0.7in, right=0.8in]{geometry}
 \usepackage{appendix}
\usepackage{graphicx}
\usepackage{authblk}
\usepackage[numbers,sort&compress]{natbib}
\usepackage{natbib}
\usepackage{pdflscape}
\usepackage{url}
\usepackage{mathrsfs}

\usepackage{amssymb}\usepackage{amsmath,amsthm}
\usepackage{amsmath} \usepackage{multirow}

\renewcommand{\theequation}{\thesection.\arabic{equation}}

\usepackage{caption}
\usepackage{subfig}
\graphicspath{ {./images/} }

\numberwithin{equation}{section} 
\usepackage{amscd,epstopdf}
\usepackage{graphicx}\usepackage{multirow}
\usepackage{float}
\usepackage{subcaption}

\usepackage{afterpage} \usepackage{rotating}
\usepackage{enumerate}
 
 \newtheorem{thm}{Theorem}[section]
 \newtheorem{proposition}{Proposition}[section]
 \newtheorem{lem}{Lemma}[section]
 \newtheorem{exam}{Example}[section]
 \newtheorem{corollary}{Corollary}[section]

 \newtheorem{definition}{Definition}[section]
 \usepackage{setspace}
\usepackage{indentfirst} 
\usepackage{graphicx}
\usepackage{bibentry}

\begin{document}

\title{An Exact Theory of Causal Emergence for Linear Stochastic Iteration Systems}

\author[1]{Kaiwei Liu}
\author[2]{Bing Yuan} 
\author[1,2]{Jiang Zhang  \thanks{Email: zhangjiang@bnu.edu.cn}   } 
 
\affil[1]{ School of Systems Science, Beijing Normal University, 100875, Beijing, China}
\affil[2]{Swarma Research, 102300, Beijing, China
}
       
 \renewcommand\Authands{ and }
 \date{}
 \maketitle
 
\begin{abstract}
After coarse-graining a complex system, the dynamics of its macro-state may exhibit more pronounced causal effects than those of its micro-state. This phenomenon, known as causal emergence, is quantified by the indicator of effective information. However, two challenges confront this theory: the absence of well-developed frameworks in continuous stochastic dynamical systems and the reliance on coarse-graining methodologies. In this study, we introduce an exact theoretic framework for causal emergence within linear stochastic iteration systems featuring continuous state spaces and Gaussian noise. Building upon this foundation, we derive an analytical expression for effective information across general dynamics and identify optimal linear coarse-graining strategies that maximize the degree of causal emergence when the dimension averaged uncertainty eliminated by coarse-graining has an upper bound. Our investigation reveals that the maximal causal emergence and the optimal coarse-graining methods are primarily determined by the principal eigenvalues and eigenvectors of the dynamic system's parameter matrix, with the latter not being unique. To validate our propositions, we apply our analytical models to three simplified physical systems, comparing the outcomes with numerical simulations, and consistently achieve congruent results. 
\end{abstract}

 \textit{ Keywords: } Causal Emergence, Effective Information, Linear Stochastic Iteration System, Coarse-graining.

\section{Introduction}

Many complex systems \cite{Holland2000} in reality, such as cities \cite{west2018scale, Li2017, Dong2020}, companies \cite{zhang2021scaling, Zhang2021, Xu2023}, bird flocks \cite{Wicks2007, Hartman2006}, perception systems \cite{Liu2023a, Wang2023, Wang}, brains \cite{Jingnan2021, Sporns2021, Varley2023}, cells \cite{Zhao2023, Dong2016}, molecules, etc., all exhibit emergent behaviors or statistical laws on a macro-level that cannot be simply derived from micro-level properties or dynamics. However, how to quantitatively characterize emergence has only recently garnered widespread attention \cite{Yuan2024,oden2003research,seth2008measuring,Rosas2020, Hoel2013, Hoel2016,rosas2024software}. The theories of causal emergence attempt to capture the conception of emergence from the point of view of causality. Intuitively, causal emergence refers to the phenomenon in dynamical systems where stronger causal effects can be obtained on macro-states by coarse-graining the micro-states \cite{Hoel2017}. There are three different ways to quantify the idea of causal emergence, which are based on effective information \cite{Hoel2013,Hoel2016}, partial information decomposition \cite{williams2010nonnegative,rassouli2019data,Rosas2020,Varley2023,rosas2024software}, and dynamical independence \cite{seth2008measuring,Sporns2021,Liu2023a}, respectively. In this paper, we mainly focus on the first way.


Erik Hoel introduced the initial quantitative theory for causal emergence, founded on effective information ($EI$) \cite{Hoel2013, Hoel2016, Hoel2017}. Nevertheless, the original framework is limited to quantifying discrete Markov chains in both the time domain and state space \cite{Liu2023, Zhang2024}, overlooking continuous state spaces. To extend the causal emergence theory in continuous spaces, P. Chvykov and E. Hoel put forth the theory of causal geometry \cite{Chvykov2020}, wherein they devised a method for calculating $EI$ in functional mappings across continuous state spaces \cite{Transtrum2015, Zhang2021}. Nonetheless, this theory solely explores general functional mappings and neglects the multi-step dynamical evolution, rendering it inapplicable to dynamical systems in continuous state spaces. Furthermore, all of Hoel's theories \cite{Hoel2013, Hoel2016, Hoel2017, Klein2020, Chvykov2020, Comolatti2022} necessitate a predefined coarse-graining strategy to discern instances of causal emergence. While this coarse-graining strategy can be derived by maximizing $EI$ for macro-dynamics, solving the optimization problem for continuous variables proves challenging \cite{Liu2021, Zhang2021, Yang2023}.

There exist various approaches to coarse-grain the micro-states of complex systems, including manually designed renormalization methods \cite{kadanoff1966scaling}, traditional dimensionality reduction techniques \cite{Villemagne1987, Boley1994, Gallivan1994, Gugercin2008, Antoulas2005}, and machine learning-based methods for coarse-graining or renormalization \cite{Liu2021, Zhang2021, Hu2020}. To automatically discover coarse-graining strategies that optimize causal emergence, J. Zhang and K. Liu introduced a machine-learning framework known as the Neural Information Squeezer (NIS) \cite{Zhang2022}, employing reversible neural networks. This framework facilitates the automatic extraction of effective coarse-graining strategies and macro-state dynamics, enabling the direct identification of causal emergence from diverse time series data. Subsequently, the research team proposed the enhanced NIS+ framework \cite{Yang2023}, which integrates the optimization of coarse-graining through $EI$ maximization into machine learning. NIS and NIS+ have successfully addressed the challenge of identifying causal emergence in data by optimizing coarse-graining techniques and macro-dynamics within continuous stochastic dynamical systems through $EI$ maximization. Nonetheless, machine learning-based methods \cite{Scholkopf2021, Murphy2023, Murphy2024} rely heavily on data adequacy and the level of neural network training. While they can provide numerical solutions, they lack a ground truth for assessing the quality of model training and the reliability of outcome indicators.

To overcome the limitations identified in previous studies, this paper will first derive analytical solutions for $EI$ and causal emergence in linear stochastic iterative systems. It will then determine the explicit expression for maximal causal emergence and the corresponding optimal coarse-graining strategies, under the condition that the dimension-averaged uncertainty eliminated by coarse-graining is constrained by a specified threshold. Subsequent to this, numerical simulation experiments will be conducted on simplified physical systems to validate the theoretical findings. Additionally, we will explore challenges such as maximizing causal emergence under an alternative condition involves constructing the coarse-graining map through a combination of dimensionality reduction projection and rotation, showcasing the interplay between determinism and degeneracy, examining the connection between maximizing causal emergence and minimizing prediction errors in dynamics, and addressing nonlinear dynamics. The primary objective of this research is to address deficiencies in $EI$ related to Transition Probability Matrices (TPM) \cite{Hoel2013, Hoel2017} and causal geometry within a toy model context \cite{Chvykov2020}. Furthermore, it aims to lay a theoretical groundwork to enhance the identification of causal emergence through machine learning \cite{Zhang2022, Yang2023}.

\section{Basic Notions and Problems Formulation}
\subsection{Continuous linear stochastic iteration system}
A stochastic iteration system \cite {Dunsmuir1976} is a sequence of random variables $x_t$ at different times $t$ which usually represent the dynamic evolution of a certain quantity, such as stock prices, temperature changes, traffic flows, etc. This paper mainly focuses on linear high-dimensional stochastic iteration systems \cite{Hannan1984} to reduce our calculation. 

For variable $x_t\in\mathcal{R}^n$ at time $t$ in the stochastic iteration system, the evolution of $x_t$ follows the stochastic iterative equation
\begin{eqnarray}\label{MicroDynamics}
	x_{t+1}=Ax_t+\varepsilon_t.
\end{eqnarray}
We define matrix $A\in\mathcal{R}^{n\times n}$ as a dynamical parameter matrix. In Equation (\ref{MicroDynamics}), $\varepsilon_t\sim\mathcal{N}(0,\Sigma)$ represents Gaussian noises with zero means and the covariance matrix $\Sigma\in\mathcal{R}^{n\times n}$. We specify the ranks of matrices ${\rm rk}(A)={\rm rk}(\Sigma)=n$. ${\rm rk}(\cdot)$ means the rank of a matrix. We can easily observe that the probability of $x_{t+1}$ falling in a region centered with $Ax_{t}$.

\subsection{Coarse-graining}
In the realm of linear stochastic iterative systems, we confine coarse-graining strategies to linear maps. Therefore, a direct approach entails employing the relatively straightforward projection matrix coarse-graining method. 

Generally, the coarse-graining strategy can be represented as $\phi(x_t): \mathbb{R}^n \to \mathbb{R}^m, m < n$. It is designed to map micro-states to a lower-dimensional space and to convert micro-states into macro-states.
\begin{definition}
	\label{dfn.Coarse-graining}
	(Coarse-graining strategy):
	For variable $x_t$ at time $t$ in the linear stochastic iteration system, we define linear mapping
	\begin{eqnarray}\label{Coarse-graining-mapping}
		\phi(x_t)=Wx_t,
	\end{eqnarray}
	 where, $W\in\mathcal{R}^{k\times n}$, ${\rm rk}(W)=k$, to map $x_t$ in high-dimensional space $\mathcal{R}^n$ to $y_t$ in low-dimensional space $\mathcal{R}^k$, achieving coarse-graining  of variables
	 \begin{eqnarray}\label{Coarse-graining-transformation}
	 	y_t=\phi(x_t)=Wx_t
	 \end{eqnarray}
	 in the stochastic iteration system.
\end{definition}

\subsection{Macro dynamical system}
Coarse-graining strategy can map a set of micro-states to specific macro-states, thus can naturally derive the expressions from micro-state parameter matrices $A$ and $\Sigma$ for dynamical models on macro-state space which can be described by macro-state parameter matrices $A_M$ and $\Sigma_M$. Since $W\in \mathcal{R}^{n\times k}$ in linear coarse-graining strategy is irreversible, so we need Moore-Penrose generalized inverse matrix \cite{Barata2012} $W^\dagger$ in matrix theory \cite{Horn2012} to derive $A_M$. 

According to Equations (\ref{Coarse-graining-transformation}) and (\ref{MicroDynamics}), we can get that
\begin{eqnarray}\label{MicrotoMacroDynamics2}
\begin{aligned}	y_{t+1}&=Wx_{t+1}=WAx_t+W\varepsilon_t.\end{aligned}
\end{eqnarray}
Due to the fact that both the micro-dynamical model itself and the coarse-graining strategies are linear transformations, and based on the properties of linear algebra, there is also a set of equations for the macro-dynamics. Thus, we hope that there is a matrix $A_M$ satisfying the following equation:
\begin{eqnarray}\label{HyMacroDynamics2}
   		\begin{aligned}
   			y_{t+1}&=A_My_t+\varepsilon_{M,t},
   		\end{aligned}
   	\end{eqnarray}
    And according to Equations (\ref{HyMacroDynamics2}), (\ref{MicrotoMacroDynamics2}) and $y_t=Wx_t$, we can get that
\begin{eqnarray}\label{MacroDynamicsPara}
   	\begin{cases}
   	    WA=A_MW,\\
            W\varepsilon_t=\varepsilon_{M,t}.
   	\end{cases}
   \end{eqnarray}
Since $\varepsilon_t\sim \mathcal{N}(0,\Sigma)$, $\varepsilon_{m,t}=W\varepsilon_t\sim \mathcal{N}(0,\Sigma_M)$. $W\in \mathcal{R}^{n\times k}$ is irreversible, so we require the Moore-Penrose generalized inverse matrix $W^\dagger\in \mathcal{R}^{n\times k}, WW^\dagger=I_k$, to solve Equation (\ref{MacroDynamicsPara}). Multiply both sides of the first equation of Equation (\ref{MacroDynamicsPara}) by $W^\dagger$, we can obtain the macro-state parameter matrix 
\begin{equation}
    A_M=WAW^\dagger,
\end{equation}
and the macro dynamics in the Definition \ref{thm.Macro-dynamical} as follows.
\begin{definition}
	\label{thm.Macro-dynamical}
	(Macro dynamics):
	For the linear stochastic iteration system in Equation (\ref{MicroDynamics}), the variables $x_t\in \mathcal{R}^n, t=0,1,\dots$, are mapped to macro state variables $y_t\in \mathcal{R}^k, t=0,1,\dots$, and there exists a new evolutionary dynamics between macro state variables, which we call macro dynamics. Macro dynamics can be represented by macro iterative equations
	\begin{eqnarray}\label{MacroDynamics}
		y_{t+1}=A_M y_t+\varepsilon_{M,t},
	\end{eqnarray}
    in which $A_M=WAW^\dagger\in \mathcal{R}^{k\times k}$, $\varepsilon_{m,t}\sim \mathcal{N}(0,\Sigma_M), \Sigma_m=W\Sigma W^{T}$. $W^\dagger$ is the Moore-Penrose generalized inverse matrix of $W$.
\end{definition}

\subsection{Inverse-coarse-graining}
To evaluate the information loss caused by coarse-graining and derive macro dynamics, we also need to transform the macro-state back to the micro-state, which is called inverse-coarse-graining. Due to the difference in dimensions $n$ and $k$ of the vectors before and after transformation, coarse-graining is not a reversible process. For linear transformations, irreversibility is manifested as $W\in \mathcal{R}^{n\times k}$ in linear coarse-graining is irreversible. 
\begin{definition}
	\label{dfn.inverse-Coarse-graining}
	(Inverse-coarse-graining):
	As $W^\dagger$ is the Moore-Penrose generalized inverse matrix of $W$, we define
	\begin{eqnarray}\label{Inverse-coarse-graining-mapping}
		\phi^\dagger(y_t)=W^\dagger y_t,
	\end{eqnarray} 
     $W^\dagger\in \mathcal{R}^{n\times k}$, ${\rm rk}(W^\dagger)=k$ as the inverse-coarse-graining mapping of variables in the stochastic iteration system. By applying the inverse-coarse-graining  
     \begin{eqnarray}\label{Inverse-coarse-graining-transformation}
     	\hat{x}_t=W^\dagger y_t,
     \end{eqnarray}
     we can map the variable $y_t \in \mathcal{R}^k$ after dimensionality reduction back to a high-dimensional space $\mathcal{R}^n$. 
\end{definition}
With inverse-coarse-graining, we can then estimate the prediction loss function of the dynamics after coarse-graining.

\section{Effective information and causal emergence}
With the definitions and expressions of micro dynamics and macro dynamics, we can study the phenomenon of causal emergence. We start with the calculation of effective information and then use the method of difference of effective information between micro-dynamics and macro-dynamics to calculate the degree of causal emergence.

\subsection{Effective information}\label{sec:EI}
We refer to Hoel's effective information ($EI$) metric to measure the magnitude of causal effects in dynamical systems. In general, $EI$ can be understood as the mutual information between the probability distributions of the states of variables at two different time points when the state at the earlier time is intervened artificially and keeps the causal mechanism, i.e., the transitional probability $Pr(x_{t+1}|x_t)$ unchanged, which is

\begin{equation}
\label{original_EI}
    EI=I(\tilde{x}_{t+1},x_t|do(x_t\sim U(X))),
\end{equation}
where $x_t\in X$, and $U(X)$ means the uniform distribution or maximum entropy distribution on $X$. The intervention is formalized using Judea Pearl's theory \cite{Pearl2009} of causality, particularly through $do(\cdot)$ operations, which means artificially defining the probability distribution space of a variable. Consequently, the distribution of $x_{t+1}$ will be indirectly altered by the intervention on $x_t$ through the causal mechanism $Pr(x_{t+1}|x_t)$. Therefore, $\tilde{x}_{t+1}$ denotes the random variable representing $x_{t+1}$ following the intervention on $x_t$. To be noticed, Equation (\ref{original_EI}) is only an operational definition and the intervention of $do$ operator is just an imaginary operation to calculate $EI$ but has no real physical meanings. 

The most classic applications of $EI$ are in the analysis of transitional probability matrix (TPM) for Markov chains. Assuming our probability transition matrix is $M=(M_{ij})_{n\times n}=(M_1^T,\dots,M_i^T,\dots,M_n^T)^T\in \mathcal{R}^{n\times n}$, $M_i=(M_{i1},\dots,M_{ij},\dots,M_{in})\in \mathcal{R}^{1\times n}$, $1\leq i,j\leq n$. Intervention on Markov transition matrix at $t$, let $P(X_{t}=x_{t})=1/n$, then the state at $t+1$ will follow the effect distribution $E_D=(\sum_i M_{i1}/n,\dots,\sum_i M_{ij}/n,\dots,\sum_i M_{in}/n)\in \mathcal{R}^{1\times n}$. So, $EI$ is \cite{Yuan2024}:
\begin{eqnarray}
\label{EI_MIJ}
    EI(M)=\frac{1}{n}\sum_{i=1}^nD_{KL}(M_i||E_D)=\frac{1}{n}\sum_{i
    ,j}M_{ij}\log_2\left(\frac{nM_{ij}}{\sum_lM_{lj}}\right),
\end{eqnarray}
where $D_{KL}(\cdot||\cdot)$ is the KL-divergence between two probability distributions. The benefit of expressing Equation (\ref{original_EI}) to the new form of Equation (\ref{EI_MIJ}) is to stress that $EI$ is only the function of the dynamics($M$) and independent of the distribution of $x_t$. This calculation method requires both variable time and state space to be discrete. 

According to previous works \cite{Hoel2013,Hoel2017,Zhang2024}, effective information $EI(M)$ can be divided into two parts, determinism $-\left<H(M_i))\right>_i$ and degeneracy $H(E_D)$, as
\begin{eqnarray}
\label{EI_MIJ_DegDet}
\begin{aligned}
EI(M)&=-\left<H(M_i))\right>_i+H(E_D)\\&=\mathop{\frac{1}{n}\sum_{i,j}M_{ij}\log_2\left(M_{ij}\right)}_{Determinism}+\mathop{\sum_{j}\left[-\frac{\sum_lM_{lj}}{n}\log_2\left(\frac{\sum_lM_{lj}}{n}\right)\right]}_{Degeneracy}.
\end{aligned}
\end{eqnarray}
Determinism measures how reliably $x_t$ leads to the future state $x_{t+1}$ of the system. Degeneracy measures to what degree there is deterministic convergence (not due to noise) from other states onto the future states $x_{t+1}$ specified by ${x_t}$, degeneracy refers to multiple ways of deterministically achieving the same effect. Two values can be calculated separately to obtain $EI$.

Further, we can generalize $EI$ to a continuous state space. We should replace summations with integrals, and change the intervened distribution of $x_t$ from a discrete equal probability distribution $P(x_{t})=1/n$ to a continuous uniform distribution $U([-L/2, L/2]^n)$ in which probability density function is $p(x_{t})=1/L^n$. Since the intervened distribution can not be defined on an infinite space such as $\mathcal{R}^n$, we set the intervention on a limited sub-region of $\mathcal{R}^n$ which is $[-L/2,L/2]^n$, where $L$ is a hyperparameter with a value being at least larger than the maximum of $x_t$ within finite time steps. 

We can refer to Hoel's framework \cite{Hoel2013,Klein2020} to derive an expression for the effective information of linear stochastic iteration system (Equation \ref{MicroDynamics}). Actually, according to Equation \ref{MicroDynamics}, the conditional distribution of $x_{t+1}$ under given $x_t$ is a Gaussian distribution, $p(x_{t+1}|x_t)=\mathcal{N}(Ax_t,\Sigma)$. Therefore, we can compute the $EI$ for this Gaussian distribution \cite{Zhang2022} as $EI = I(x_{t+1}, x_t | do(x_t \sim U))$, drawing an analogy to mutual information and information entropy within a Gaussian distribution \cite{Venkatesh2024} as
\begin{eqnarray}\label{originEI}
		EI(A,\Sigma)=\ln\displaystyle\frac{|\det(A)|L^n}{(2\pi e)^\frac{n}{2}\displaystyle \det(\Sigma)^\frac{1}{2}}.
\end{eqnarray}
Similar to $EI(M)$ of discrete systems, $EI(A,\Sigma)$ of continuous systems can also be decomposed into two terms: determinism $-\left<H(p(x_{t+1}|do(x_t)))\right>$ and degeneracy $H(E_D(x_{t+1}))$ as
\begin{eqnarray}\label{originEIdegdet}
    \begin{aligned}
		EI(A,\Sigma)&=-\left<H(p(x_{t+1}|x_t))\right>+H(E_D(x_{t+1}))\\&=\mathop{-\ln\left[(2\pi e)^\frac{n}{2}\det(\Sigma)^\frac{1}{2}\right]}_{Determinism}+\mathop{\ln\left(|det(A)|L^n\right)}_{Degeneracy}.
    \end{aligned}
\end{eqnarray}
The higher the determinism or the lower the degeneracy, the stronger the causal effect of the stochastic iteration system \cite{Hoel2013}. The specific derivation process can refer to the calculation process in \ref{appendixEIandce}. 

Since the calculation of $EI$ involves multiple integrals, $L^n$ appears in the expression of $EI$. When the value of $L$ is high, the value of $EI$ becomes higher after $n$ powers of $L$. So we stipulate that when calculating causal emergence, the $EI$ in different dimensions is dimensionally averaged to eliminate the impact of power growth of $L$. Another reason for taking the average on the dimension is that for the comparison of $EI$ from different dimensions, dimension averaged $EI$ can eliminate $L$, which is shown in Section \ref{sec.ce}. Here, in Theorem \ref{thm.Effective-information}, we define a new index that is effective information for stochastic iterating systems ($\mathcal{J}$).
\begin{definition}
	\label{thm.Effective-information}
	(Effective information for stochastic iterating systems):
	For the linear stochastic iteration systems like Equation (\ref{MicroDynamics}), the effective information of the dynamical system is calculated as
	\begin{eqnarray}\label{EI}
		\mathcal{J}(A,\Sigma)\equiv \frac{EI(A,\Sigma)}{n}=\frac{1}{n}\ln\displaystyle\frac{|\det(A)|L^n}{(2\pi e)^\frac{n}{2}\displaystyle \det(\Sigma)^\frac{1}{2}}=\ln\displaystyle\frac{|\det(A)|^\frac{1}{n}L}{(2\pi e)^\frac{1}{2}\displaystyle \det(\Sigma)^\frac{1}{2n}}
	\end{eqnarray}
    $\det(\cdot)$ represents the determinant value corresponding to a matrix $\cdot$. $|\cdot|$ represents the absolute value of $\cdot$. $L$ represents the interval size of $do(\cdot)$ intervention. 
    
    The effective information $\mathcal{J}$ can also be decomposed into two parts, determinism and degeneracy as
    \begin{eqnarray}\label{EIDeterminismDegeneracy}
    \begin{aligned}
    	\mathcal{J}(A,\Sigma)
     &=\mathop{-\ln\left[(2\pi e)^\frac{1}{2}\det(\Sigma)^\frac{1}{2n}\right]}_{Determinism}+\mathop{\ln\left(|det(A)|^\frac{1}{n}L\right)}_{Degeneracy}.
     \end{aligned}
    \end{eqnarray}
\end{definition}
The proof of Theorem \ref{thm.Effective-information} can refer to \ref{appendixEIandce}. $\mathcal{J}(A,\Sigma)$ can be used to measure the causal effects of micro dynamics and macro dynamics then calculate causal emergence. To be noticed, the scope $L$ of $do$ intervention in macro-state dynamics $EI$ is related to $W\in \mathcal{R}^{n\times k}$. So we need to limit the range of $W$ to ensure the $do$ intervention range of macro-state $y_t\in\mathcal{R}^k$ and micro-state $x_t\in\mathcal{R}^n$ as $do(y_t)\sim U([-L/2,L/2]^k)$ and $do(x_t)\sim U([-L/2,L/2]^n)$, which we will conduct more analysis in Subsection \ref{sec:Maximizing-causal-emergence}. 

\subsection{Causal emergence}\label{sec.ce}
After obtaining effective information for stochastic iterating systems $\mathcal{J}$ at different dimensions $k$ of variables $x_t$, we still have a problem to solve. Except for the influence of indices, the value of hyperparameter $L$ itself, which is artificially assumed, also has a great influence on the results of $\mathcal{J}$. This hyperparameter can be eliminated by calculating causal emergence for stochastic iterating systems as Definition \ref{JMsubJm}.
\begin{definition}\label{JMsubJm}
(Causal emergence for stochastic iterating systems) Hyperparameter $L$ can be eliminated by subtraction as
\begin{eqnarray}\label{dimensionally_averaged}
	\Delta \mathcal{J} = \mathcal{J}_M-\mathcal{J}_m,
\end{eqnarray}
where $\mathcal{J}_m\equiv \mathcal{J}(A,\Sigma)$ is the $EI$ for the micro-dynamics (Equation \ref{MicroDynamics}, and $\mathcal{J}_M\equiv \mathcal{J}(A_M,\Sigma_M)$ is the $EI$ for the macro-dynamics (Equation \ref{MacroDynamics}), and $\Delta \mathcal{J}$ is the degree of causal emergence for stochastic iterating systems. 
\end{definition}
Through the calculations in the \ref{appendixEIandce}, we obtain the following theorem. 
In this theorem, we not only calculate the analytical solution for causal emergence but also ensure that the degree of the analytical solution is only related to the parameters of stochastic iterating systems without the influence of $L$.

\begin{thm}
	\label{thm.Causal-emergence}
	(Analytical solution for causal emergence):
	For the linear stochastic iteration systems like Equation (\ref{MicroDynamics}), causal emergence of the system after coarse-graining  $y_t=\phi(x_t)=Wx_t$, $W\in \mathcal{R}^{k\times n}$, is calculated as 
	\begin{eqnarray}\label{Causal-emergence}
		\Delta\mathcal{J}=\mathop{\ln\frac{|\det(WAW^\dagger)|^\frac{1}{k}}{|\det(A)|^\frac{1}{n}}}_{Degeneracy Emergence}+\mathop{\ln\frac{|\det(\Sigma)|^\frac{1}{2n}}{|\det(W\Sigma W^{T})|^\frac{1}{2k}}}_{Determinism Emergence}
	\end{eqnarray}
    $W\in \mathcal{R}^{k\times n}$, $x_t\in \mathcal{R}^{n}$, $y_t\in \mathcal{R}^{k}$. Causal Emergence can be also divided into two terms, and we name them as degeneracy emergence and determinism emergence, separately. Both terms depend on the coarse-grained results of the overall derivative and noise indicators, respectively.
\end{thm}
The premise for the Theorem \ref{thm.Causal-emergence} to hold is that both macro and micro states can be represented by $L$ as the range of values, and the premise is that the two cannot differ by too much order of magnitude. When the intervention interval size for each dimension of micro and macro variables is consistent and equal to $L$, $L$ can be directly eliminated when calculating causal emergence, so $\Delta\mathcal{J}$ is independent with $L$.

\subsection{Some Bounds related with Causal Emergence}
After inferring the calculation method of $\Delta\mathcal{J}$, we can determine whether there is causal emergence in the macro states after coarse-graining. Next, we will find an optimal matrix $W$ to maximize the degree of causal emergence $\Delta\mathcal{J}$ \cite{Hoel2013,Zhang2022}. To solve this problem, we will firstly discuss the upper bounds related with the two terms, the determinism emergence and degeneracy emergence, in Equation \ref{Causal-emergence}, separately. Here we have lemmas below.

The first lemma is to analyze degeneracy emergence by optimizing $A_M=WAW^\dagger$.
\begin{lem}
	\label{thm.Maximizing-determinant}
For matrix $A$, its eigenvalues can be real numbers or complex numbers, while the range of absolute values(modulus) for $|\det(WAW^\dagger)|$ is bounded:
	\begin{eqnarray}\label{Maximizing-determinant}
		0\leq\left|\det(WAW^\dagger)\right|\leq\prod_{i=1}^{k}|\lambda_i|,
	\end{eqnarray}
	where, $A\in \mathcal{R}^{n\times n}$, $W\in \mathcal{R}^{k\times n}$, $W^\dagger\in \mathcal{R}^{n\times k}$ is the Moore-Penrose generalized inverse matrix of $W$, and $k$ is the rank of the matrix $W$. 
	\begin{enumerate}[(1)]
	\item When the eigenvalues are all real numbers, $|\lambda_1|\geq|\lambda_2|\geq\dots\geq|\lambda_n|\geq 0$ are $n$ absolute values of eigenvalues of matrix $A$ sorted from largest to smallest, $\lambda_1,\lambda_2,\dots,\lambda_n\in R$.
	\item When the eigenvalues contain imaginary numbers, $|\cdot|$ represents the modulus of complex numbers. Note that since the eigenvalues of imaginary numbers appear together with their conjugate complex numbers, both eigenvalues must have equal moduli and must be discarded or retained simultaneously.
	\item When $A$ is a symmetric matrix and $A>0$ is a positive definite matrix, since $\lambda_1\geq \lambda_2\geq\dots\geq \lambda_n\geq0$, we can just write the equation as
	\begin{eqnarray}\label{Maximizing-determinant_symmetric}
		\displaystyle\prod_{i=n-k+1}^{n}\lambda_i\leq\det(WAW^\dagger)\leq\prod_{i=1}^{k}\lambda_i.
	\end{eqnarray}
	\end{enumerate}
\end{lem}
The proof of the lemma can be referred to in \ref{appendixAM}. In addition, we can verify this lemma through numerical simulations. For example, for matrix $A$ with eigenvalues of $\lambda_1=2.540, \lambda_2=1.380, \lambda_3=-0.4899, \lambda_4=0.1149$, we randomly generate the elements of the matrix $W$. From Fig.\ref{fig:generateW}a We can see that there is an upper bound $\lambda_1\lambda_2$ on the corresponding $\det(A_M)$ value of each $W$. If $A$ is also randomly generated, we can see that all the scatter points are below the diagonal line $y=x$ in Fig.\ref{fig:generateW}b which means Lemma \ref{thm.Maximizing-determinant} is satisfied.

After the lemma about $|\det(A_M)|=|\det(WAW^\dagger)|$, the next lemma is about $\det(\Sigma_M)=\det(W\Sigma W^T)$, which is related with the determinism emergence.

\begin{lem}
     \label{thm.SigmadeterminantandSW}
     We define the singular values of $W$ are $s_{W,1}\geq s_{W,2}\geq \dots \geq s_{W,k}$, then $\det(W\Sigma W^{T})$ satisfies
\begin{eqnarray}
	\label{eq:singulars}
 \left(\prod_{i=1}^{k}s_{W,i}\kappa^\frac{1}{2}_{n-i+1}\right)^\frac{1}{k}\leq\det(W\Sigma W^{T})^\frac{1}{2k}\leq\left(\prod_{i=1}^{k}s_{W,i}\kappa^\frac{1}{2}_{i}\right)^\frac{1}{k},
\end{eqnarray}
$\kappa_1\geq\kappa_2\geq\dots\geq\kappa_n$ are the eigenvalues of $\Sigma$. When $\Sigma=\sigma^2I$, the equal sign holds and
\begin{eqnarray}
   \det(W\Sigma W^{T})^\frac{1}{2k}=\sigma\left(\prod_{i=1}^{k}s_{W,i}\right)^\frac{1}{k}.
\end{eqnarray}
\end{lem}
In Fig.\ref{fig:generateW}c, we validate the inequalities using randomly generated $\Sigma$ and $W$. Therefore, we can control the range of $\det(W\Sigma W^{T})$ by adjusting the magnitude of the singular values of $W$. We further can adjust the magnitude of the singular value of $W$ to weaken the noise and enhance the causal emergence as we know in Lemma \ref{thm.SigmadeterminantandSW}. As shown in Fig.\ref{fig:generateW}d, the smaller the mean singular values of $W$ we generate, the greater the degree of causal emergence $\Delta\mathcal{J}$.

\subsection{Causal Emergence Maximization under Given Information Loss}\label{sec:Maximizing-causal-emergence}
In this subsection, we will officially discuss the problem of causal emergence maximization. Before proceeding, we will introduce a condition on uncertainty elimination to constrain $W$, such that zeros in denominator are excluded. 

\subsubsection{The Constraint on the uncertainty elimination by coarse-graining}

According to Lemma \ref{thm.SigmadeterminantandSW}, when the singular value of $W$, $s_{W,1}\to 0$, $\det(W\Sigma W^T)\to 0$. This means if $\phi(x_t)=Wx_t$ is a completely unconstrained map, the values of $x_t$ and the noise $\epsilon_t$ may be scaled in an unbounded manner and the information entropy $H(y_t)\to-\infty$, leading to $y_t$ having an excessively narrow range. To prevent the macro variables from being excessively compressed by the coarse-graining map which may result in excessive information content, it is essential to impose the following restriction,
\begin{eqnarray}\label{SetWSigEntropy}
	\frac{1}{n}H(p(x_t))-\frac{1}{k}H(p(y_t))\leq\eta,
\end{eqnarray}
where, $n$ and $k$ are the dimensionalities of micro- and macro- variables, and $\eta>0$ is a given constant, and $H(p(x_t))=(n/2)\ln(2\pi e)+(1/2)\ln(\det(\Sigma))$ and $H(p(y_t))=(k/2)\ln(2\pi e)+(1/2)\ln(\det(W\Sigma W^T))$. The inequality can be understood as the dimension averaged uncertainty eliminated (or information gained) by the coarse-graining strategy is under a given threshold $\eta$. According to the inequality (\ref{SetWSigEntropy}), we can further get that
\begin{eqnarray}\label{SetWSig}
	\det(W\Sigma W^{T})^\frac{1}{k}\geq\epsilon\det(\Sigma)^\frac{1}{n},
\end{eqnarray}
where, $W^{T}$ represents the transpose matrix of matrix $W$ and $\epsilon=\exp(-2\eta)$. The covariance matrix affects the range of values of random variables, the limitations here can also ensure the effectiveness of the $do$ intervention scope $L$ mentioned in Section \ref{sec:EI}. It will be clear in the subsequent sections that this constraint could also significantly influence the scope of causal emergence. 

With the help of Lemmas \ref{thm.Maximizing-determinant} and \ref{thm.SigmadeterminantandSW}, we can obtain the optimal solution for causal emergence $\Delta\mathcal{J}^{*}$.

\begin{figure}[htbp]
	\centering
	\begin{minipage}[c]{0.4\textwidth}
		\centering
		\includegraphics[width=1\textwidth]{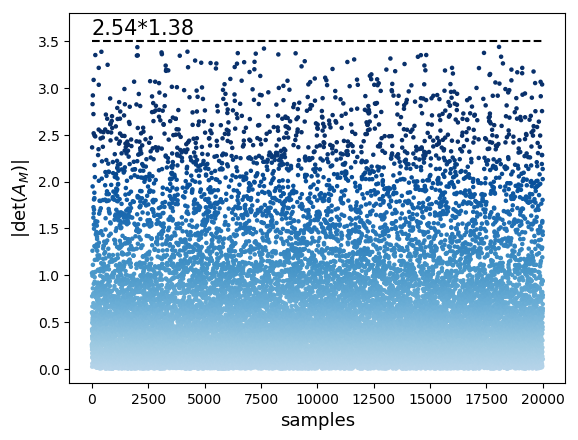}
		\centerline{(a)}
	\end{minipage}
	\begin{minipage}[c]{0.4\textwidth}
		\centering
		\includegraphics[width=1\textwidth]{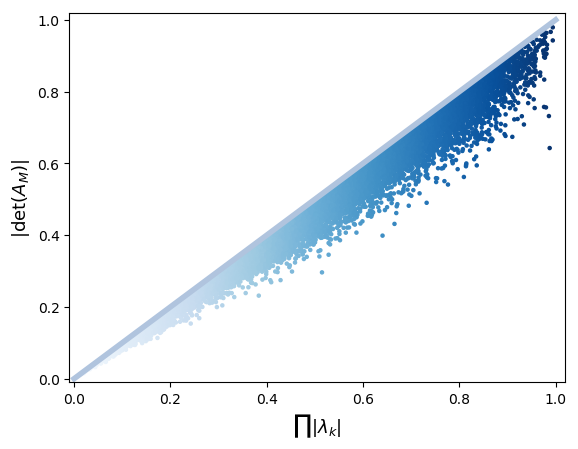}
		\centerline{(b)}
	\end{minipage}
	\begin{minipage}[c]{0.4\textwidth}
		\centering
		\includegraphics[width=1\textwidth]{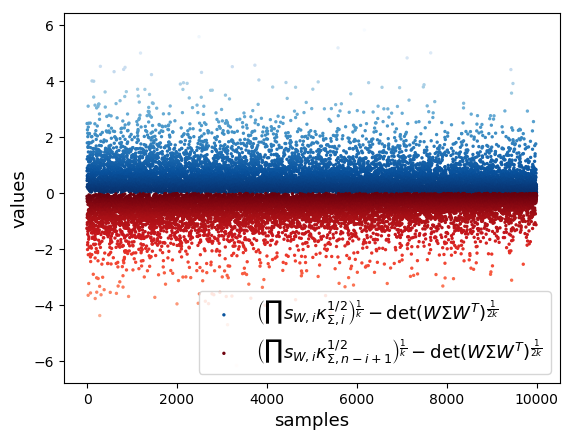}
		\centerline{(c)}
	\end{minipage}
        \begin{minipage}[c]{0.4\textwidth}
		\centering
		\includegraphics[width=1\textwidth]{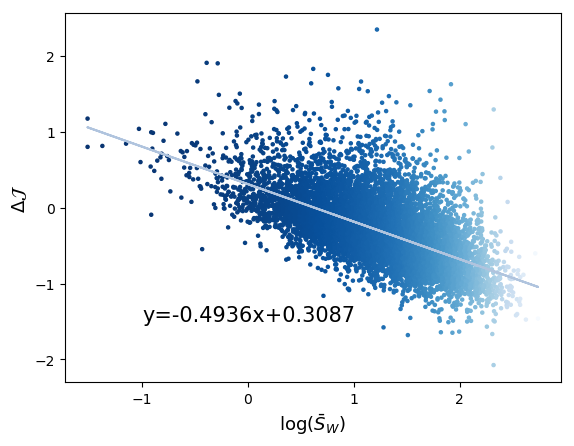}
		\centerline{(d)}
	\end{minipage}
	\caption{Experimental results of the simulation about the lemma and the theorem presented in Section 3. (a) For matrix $A$ with eigenvalues of $\lambda_1=2.540, \lambda_2=1.380, \lambda_3=-0.4899, \lambda_4=0.1149$, we randomly generate matrix $W$. We can see that there is an upper bound $\lambda_1\lambda_2$ on the corresponding $\det(A_M)$ value of each $W$. (b) If $A$ is also randomly generated, we can see the scatter plot that conforms to Lemma.\ref{thm.Maximizing-determinant}. (c) We validate the inequality using randomly generated $\Sigma$ and $W$. We can control the range of $\det(W\Sigma W^{T})$ by adjusting the magnitude of the singular value of $W$. (d) The smaller the mean singular value of $W$ we generate, the greater the degree of causal emergence $\Delta\mathcal{J}$.}
	\label{fig:generateW}
\end{figure}

\begin{thm}
	\label{thm.Maximizing-EI}
	For the linear stochastic iteration systems like Equation (\ref{MicroDynamics}), after coarse-graining $y_t=\phi(x_t)=Wx_t$, $W\in \mathcal{R}^{k\times n}$, under the constraint of Equation (\ref{SetWSig}), the maximum of the degree of causal emergence that the system can achieve is
	\begin{eqnarray}\label{maxCE}
	\Delta\mathcal{J}^{*}=\frac{1}{k}\sum_{i=1}^{k}\ln\displaystyle|\lambda_i|-\frac{1}{n}\sum_{i=1}^{n}\ln\displaystyle|\lambda_i|+\eta
	\end{eqnarray}
    $|\lambda_1|\geq|\lambda_2|\geq\dots\geq|\lambda_n|\geq 0$ are $n$ modulus of eigenvalues of matrix $A$ sorted from the largest to the smallest, $\lambda_1,\lambda_2,\dots,\lambda_n\in C$.
\end{thm}
\begin{proof}
	We will discuss the maximization of causal emergence from determinism emergence and degeneracy emergence, separately.
\begin{enumerate}[(1)]
	\item Degeneracy emergence ($\Delta\mathcal{J}_1$)
	
	Based on Theorem \ref{thm.Maximizing-determinant}, we know that
	\begin{eqnarray}\label{maxDegeneracy}
		\begin{aligned}
\Delta\mathcal{J}_1&=\ln\frac{|\det(WAW^\dagger)|^\frac{1}{k}}{|\det(A)|^\frac{1}{n}}\leq\ln\frac{\left(\displaystyle\prod_{i=1}^k|\lambda_i|\right)^\frac{1}{k}}{\left(\displaystyle\prod_{i=1}^n|\lambda_i|\right)^\frac{1}{n}}=\frac{1}{k}\sum_{i=1}^{k}\ln\displaystyle|\lambda_i|-\frac{1}{n}\sum_{i=1}^{n}\ln\displaystyle|\lambda_i|.
		\end{aligned}
	\end{eqnarray}
    If and only if the maximum $k$ eigenvalues are preserved, the equality holds. At this point, $A_M=WAW^\dagger=(WV)\Lambda(WV)^\dagger=\tilde{W}\Lambda\tilde{W}^\dagger$, $\Lambda={\rm diag}(\lambda_1, \dots, \lambda_n)\in \mathcal{R}^{n\times n}$, it is satisfied that $\tilde{W}=(\tilde{W}_k,O_{k\times{(n-k)}})$, $\tilde{W}_k\in \mathcal{R}^{k\times k}$ can be any invertible matrix. $V=(v_1,\dots,v_n)$, $v_i$ is the eigenvector corresponding to the eigenvalues $\lambda_i$ of matrix $A$, $i=1,\dots,n$, $\lambda_1\geq\dots\geq\lambda_n$. Since $\tilde{W}=WV$, We can obtain the expression for $W$ as
    \begin{eqnarray}\label{W*}
    		W=(\tilde{W}_k,O_{k\times{(n-k)}})V^{-1}
    \end{eqnarray}
	\item Determinism emergence ($\Delta\mathcal{J}_2$)

	Based on Equation (\ref{Causal-emergence}), and the inequalities (\ref{eq:singulars}), (\ref{SetWSigEntropy}) and (\ref{SetWSig}), we know that 
	\begin{eqnarray}\label{maxdeterminism}
		\begin{aligned}
				\Delta\mathcal{J}_2&=\ln\frac{|\det(\Sigma)|^\frac{1}{2n}}{|\det(W\Sigma W^{T})|^\frac{1}{2k}}=\frac{1}{2}\ln\frac{|\det(\Sigma)|^\frac{1}{n}}{|\det(W\Sigma W^{T})|^\frac{1}{k}}\\
			&\leq\frac{1}{2}\ln\frac{|\det(\Sigma)|^\frac{1}{n}}{|\det(\Sigma)|^\frac{1}{n}\epsilon}=\frac{1}{2}\ln\frac{1}{\epsilon}=\eta
		\end{aligned}
	\end{eqnarray}
when the information loss reaches the upper bound, the determinism emergence reaches the maximum. Determinism is enhanced while ensuring that randomness is not excessively reduced.

\end{enumerate}
	By taking the maximum value of $\Delta\mathcal{J}_1$ and $\Delta\mathcal{J}_2$, we can obtain
	\begin{eqnarray}\label{maxCE12}
		\begin{aligned}
			\Delta\mathcal{J}^{*}&=\Delta\mathcal{J}_1^*+\Delta\mathcal{J}_2^*=\ln\frac{\left(\displaystyle\prod_{i=1}^k|\lambda_i|\right)^\frac{1}{k}}{\left(\displaystyle\prod_{i=1}^n|\lambda_i|\right)^\frac{1}{n}}+\frac{1}{2}\ln\frac{1}{\epsilon}=\frac{1}{k}\sum_{i=1}^{k}\ln\displaystyle|\lambda_i|-\frac{1}{n}\sum_{i=1}^{n}\ln\displaystyle|\lambda_i|+\eta
		\end{aligned}
	\end{eqnarray}
\end{proof}

Suppose the solution set of the optimal coarse-grained matrix $W^*$ which satisfying $\Delta\mathcal{J}(W^*)=	\Delta\mathcal{J}^*$, is $\mathcal{W}^*$. Although it is difficult for us to find all the elements in set $\mathcal{W}^*$, we will give methods for calculating special solutions for different types of matrices in later sections.
\subsubsection{Optimal solution set $\mathcal{W}^*$}
After optimizing determinism emergence and degeneracy emergence, we can find two solution sets corresponding to the optimal solutions. Thus, the intersection of the two solution sets is the solution set of $W$ corresponding to the maximum degree of causal emergence.
\begin{corollary}
	\label{thm.setofW}
 When $\varepsilon_t\sim N_\mathcal{N}(0,\Sigma)$ and the degree of causal emergence reaches its maximum value, $W$ needs to satisfy
    \begin{eqnarray}\label{WSimas0}
		\begin{cases}
  WV=(\tilde{W}_k,O_{k\times{(n-k)}}),\\
  \det{(W\Sigma W^{T})}^\frac{1}{k}=\epsilon\det{(\Sigma)}^\frac{1}{n},
        \end{cases}
    \end{eqnarray}
	in which $V=(v_1,\dots,v_n)$, $v_i$ is the eigenvector corresponding to the eigenvalues $\lambda_i$ of matrix $A$, $i=1,\dots,n$, $\lambda_1\geq\dots\geq\lambda_n$. $\tilde{W}_k\in \mathcal{R}^{k\times k}$ can be any invertible matrix. The solution set of the optimal coarse-grained matrix $W^*$ satisfying $\Delta\mathcal{J} (W^*)=\Delta\mathcal{J}^*$ as $\epsilon=\exp(-2\eta)$ is
	\begin{eqnarray}\label{WSimas1}
		\mathcal{W}^*=\left\{W\bigg| W=(\tilde{W}_k,O_{k\times{(n-k)}})V^{-1}, \det(W\Sigma W^{T})=\epsilon\det(\Sigma)^\frac{k}{n}\right\}.
	\end{eqnarray}
\end{corollary}
We can find the proof process in the previous Subsection \ref{sec:Maximizing-causal-emergence}.

\subsubsection{A case study of $\mathcal{W}^*$}
To understand the physical meaning of the solution set, we can use a simple example in the three-dimensional space to visualize our optimal solution set. In specific cases, when $k=2, n=3$, the projection of $W$ in three-dimensional space is a circle. In this case, matrix $W=(w_1^{T},w_2^{T})^{T}\in \mathcal{R}^{2\times 3}$ can be split into two row-vectors, $w_i=(w_{i1},w_{i2},w_{i3}),i=1,2$, in $\mathcal{R}^3$, The eigenvector matrix $V=(v_1,v_2,v_3)\in \mathcal{R}^{3\times 3}$ can be regarded as the combining of three vectors, $v_j=(v_{1j},v_{2j},v_{3j})^T\in \mathcal{R}^3,j=1,2,3$. To visualize the solution set of $W$, we introduce the following case as $\epsilon=\exp(-2\eta)$.
\begin{exam}
	\label{thm.setofW3d}
In order to present the results more intuitively, we specify $w_1w_2^{T}=0$, $\Sigma=\sigma^2I_3$. According to Equation (\ref{WSimas1}), we can get
\begin{eqnarray}\label{Wsetelement}
		\begin{cases}
w_1v_3=0,\\
w_2v_3=0,\\
  (w_1w_1^{T})(w_2w_2^{T})=\epsilon,
        \end{cases}
\end{eqnarray}
 Although $W\in \mathcal{R}^{2\times 3}$ is in a six-dimensional space, we can get the solution set of $w_1$ for specified $w_2$
    \begin{eqnarray}\label{WSimas0R3}
		\begin{cases}
w_{11}v_{13}+w_{12}v_{23}+w_{13}v_{33}=0,\\
  w_{11}^2+w_{12}^2+w_{13}^2=R^2,
        \end{cases}
    \end{eqnarray}
in which $R^2=\epsilon/(w_2w_2^{T})$, $R$ is the radius of the sphere. Similarly, we can also obtain the range of values for $w_2$. 
\end{exam}
\begin{proof}
To meet the conditions $WV=(Wv_1,Wv_2,Wv_3)=(\tilde{W}_k,O_{k\times{(n-k)}})$, we only need to specify $w_1v_3=w2v_3=0$. This means that $w_1$ and $w_2$ are contained in the plane perpendicular to $v_3$ and passing through the origin, which is $ w_{11}v_{13}+w_{12}v_{23}+w_{13}v_{33}=0$. When $w_1$ and $w_2$ are in this plane, the degeneracy emergence $\Delta\mathcal{J}_1$ reaches its minimum value.

For determinism emergence $\Delta\mathcal{J}_2$,
 \begin{eqnarray}\label{WSimasr3}
 \begin{aligned}
     \det{(W\Sigma W^{T})}&=(w_1\Sigma w_1^{T})(w_2\Sigma w_2^{T})-(w_1\Sigma w_2^{T})(w_2\Sigma w_1^{T})\\
     _{(\Sigma=\sigma^2 I_3)}&=\sigma^4\left[(w_1 w_1^{T})(w_2 w_2^{T})-(w_1 w_2^{T})(w_2 w_1^{T})\right]\\
     _{w_1w_2^{T}=0}&=\sigma^4\left[(w_1 w_1^{T})(w_2 w_2^{T})\right]
    =\epsilon\det{(\Sigma)}^\frac{k}{n} \\
    &= \epsilon\sigma^4,
 \end{aligned}
\end{eqnarray}
When $w_2$ is fixed, the solution set of $w_1$ is a sphere in three-dimensional space $w_{11}^2+w_{12}^2+w_{13}^2=\mathcal{R}^2,\mathcal{R}^2=\epsilon/(w_2w_2^{T})$. If we fix $w_1$, the result is the same. 
\end{proof}
From this example, we can derive the following proposition
\begin{proposition}
In three-dimensional space $\mathcal{R}^3$, if the noise is a white noise sequence $\Sigma=\sigma^2I_3$ and $w_1$ and $w_2$ are perpendicular to each other, when causal emergence $\Delta\mathcal{J}(W^*)=\Delta\mathcal{J}^*$, the solution set of $w_i$ for $i=1,2$ is the intersection of a plane and a sphere in the space, which is a circle.
\end{proposition} 

In Fig.\ref{fig:generate3d}, we visualize the solution set of $w_1=(w_{11},w_{12},w_{13})$. At the same time, we can extend the scope of application of Example \ref{thm.setofW3d} for general form of covariance matrix $\Sigma$, the solution set becomes an ellipsoid. If $w_1w_2^{T}\neq 0$, then the spherical or ellipsoidal surface undergoes translation. Therefore, the final solution sets for $w_1$ and $w_2$ are the intersection lines of ellipsoids or spheres with planes. Two spatial curves form the solution set of $W$. So we can get the proposition below.
\begin{proposition}
 In three-dimensional space $\mathcal{R}^3$, $\Sigma>0$ the coarse-graining parameter matrix $W$ can be regarded as composed of $k$ row vectors $w_i$, $k=1,2,3$, $i=1,\dots,k$, the solution space of each vector $w_i$ is the intersection of a plane and an ellipsoid, which is an ellipse.
\end{proposition}

With a visualized solution set in three-dimensional space, we can infer the optimal solution of $W$ in high-dimensional space.
\begin{proposition}
For any $n>1$, a coarse-graining matrix $W\in \mathcal{R}^{n\times k}$ can be split into $k$ row vectors $w_i\in \mathcal{R}^n$ in the same way, and the solution set of $\mathcal{W}^*$ can be understood as the intersection of hyperplanes and hyperellipsoids in an $\mathcal{R}^n$.
\end{proposition}

So after obtaining the relevant parameters of dynamics, we first derive the expression for a high-dimensional ellipse, where any point on the high-dimensional ellipse can maximize the causal emergence. At the same time, when the model is more complex, we can also choose to find points close to the ellipse, even if it is not the optimal solution, which can increase the likelihood of causal emergence.

\begin{figure}[htbp]
    \centering
    \includegraphics[width=0.6\textwidth]{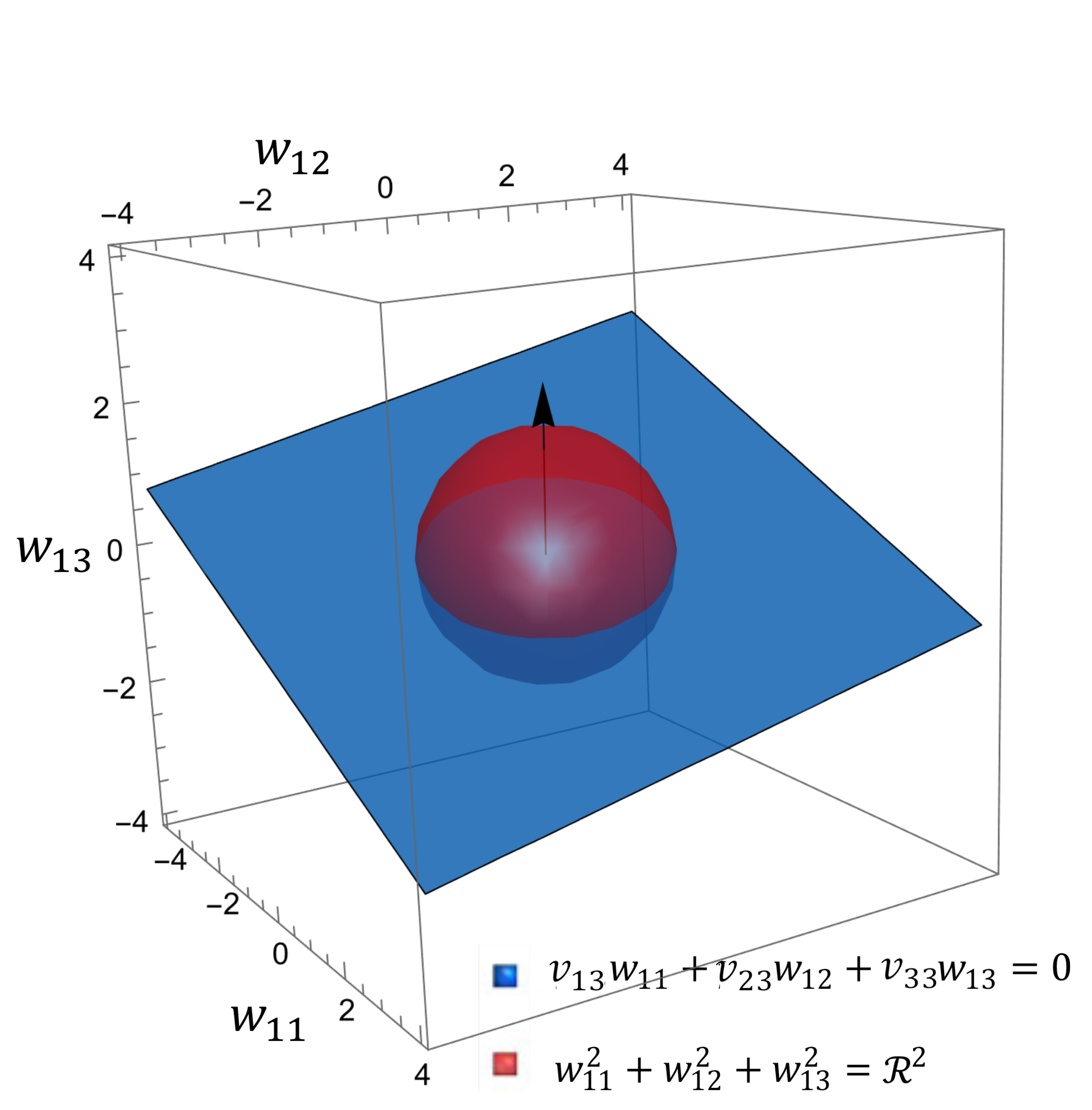}
    \caption{The visualization of the set for the optimal solutions of the coase-graining strategy $W$. When $w_1w_2^{T}=0$, $\Sigma=\sigma^2I_3$. Although $W\in \mathcal{R}^{2\times 3}$ is in a six-dimensional space, we can draw the range of values for $w_1$ while limiting $w_2$, $w_{11}v_{13}+w_{12}v_{23}+w_{13}v_{33}=0$ and $w_{11}^2+w_{12}^2+w_{13}^2=R^2,$
in which $R^2=\epsilon/(w_2w_2^{T})$. When causal emergence $\Delta\mathcal{J}(W^*)=\Delta\mathcal{J}^*$, the solution set of $w_i$ is the intersection of a plane (blue) and a sphere (red) in 3D space $\mathcal{R}^3$, which is a circle}
    \label{fig:generate3d}
\end{figure}
\subsection{The conditions for causal emergence}
After knowing the solution method for the maximum causal emergence of the system, we need to analyze the conditions under which causal emergence occurs as $\Delta\mathcal{J}>0$. Two scenarios are being examined: the presence of causal emergence in typical situations and its manifestation under optimized conditions.

First, in general cases, causal emergence $\Delta\mathcal{J}$ is determined by two factors: degeneracy $\Delta\mathcal{J}_1$ and determinism $\Delta\mathcal{J}_2$. The sufficient condition for $\Delta\mathcal{J}>0$ is that $\Delta\mathcal{J}_1,\Delta\mathcal{J}_2>0$ at the same time.

The prerequisite for $\Delta\mathcal{J}_1>0$ is that there is a significant difference between the eigenvalues $\lambda_1\geq\dots\geq\lambda_n$ of parameter matrix $A\in\mathcal{R}^{n\times n}$. Assuming the macro-state $y_t\in \mathcal{R}^{k}$, the greater the difference between $\lambda_1,\dots,\lambda_k$ and $\lambda_{k+1},\dots,\lambda_{n}$, the more obvious the causal emergence. The best condition is that $\lambda_k\gg\lambda_{k+1}$, if the latter $n-k$ term is discarded, there will be a significant causal emergence.

The condition for $\Delta\mathcal{J}_2>0$ is that there is significant randomness in the system as $\Sigma$ is a positive definite matrix and $H(p(x_t))>-\infty$. If $\Sigma=0$ and $H(p(x_t))=-\infty$ as there is no random noise in the system, effective information is also invalid. In theory, the larger the system noise, the larger the space that can be compressed during coarse-graining, and the stronger the causal emergence can appear.

Second, when causal emergence reaches its optimal solution, we can also determine the condition for $\Delta\mathcal{J}>0$ based on the result of $\Delta\mathcal{J}^*$. From Equation (\ref{maxCE12}), we can know that when 
\begin{eqnarray}\label{CEgeq0}
		\begin{aligned}
\frac{1}{n}\sum_{i=1}^{n}\ln\displaystyle|\lambda_i|-\frac{1}{k}\sum_{i=1}^{k}\ln\displaystyle|\lambda_i|<\eta
		\end{aligned}
	\end{eqnarray}
as the difference of the average singular value of $A$ to the average singular value of $A_M$ after taking the logarithm can be smaller than the bound of the information entropy $\eta$ we specify, $\Delta\mathcal{J}^*>0$. If this result holds, we can assert that there must exist a $W$ to make $\Delta\mathcal{J}>0$.
\section{Results}

After knowing how to optimize causal emergence, we can use the theorems we derived to analyze several cases of linear stochastic iteration systems. We will attempt to search for the maximum of causal emergence in systems with known dynamics. We provide three cases, focusing on the applications of determinism emergence, degeneracy emergence, and the manifestation of causal emergence in $\mathcal{R}^3$ space.

\subsection{Random walk}\label{secRandomWalk}
Our first case is random walk \cite{Lawler2010}, and our analysis focuses on the noise $\varepsilon_t$ and the covariance matrix $\Sigma$, which mainly impact determinism emergence. The random walk model is a mathematical model used to describe the random movement of objects in a specific space, in which, a walker moves between a series of locations, with the direction and distance of each movement being random. This model can study various phenomena, such as changes in asset prices in financial markets, the diffusion of particles in fluids, etc. In the random walk model, the walker can only move a certain distance to the left or right each time, and the distance obeys a normal distribution. When multiple wanderers coexist, a changing parameter matrix can be formed as an identity matrix, $A=I_n$,
\begin{eqnarray}\label{Random}
	x_{t+1}=x_t+\varepsilon_t, \varepsilon_t\sim \mathcal{N}(0,\Sigma),
\end{eqnarray}
which forms a linear stochastic series mainly affected by noise sequences $\varepsilon_t$.

We set an example as $n=4$, $k=1$ while
\begin{eqnarray}\label{RandomSigma}
	\begin{gathered}
	\Sigma=\begin{pmatrix} 0.4782& -0.1967& -0.0287&  0.0419\\
       -0.1967&  0.6711& 0.0233& -0.1067\\
       -0.0287&  0.0233&  0.3154&  0.0738\\
        0.0419&-0.1067& 0.0738&  0.4211\end{pmatrix}.
\end{gathered}
\end{eqnarray}
The motion trend of each dimension of $x_t$ is shown in Fig.\ref{fig:random_walk}a. From Equation (\ref{SetWSig}) we already know that the ratio between $\det{(W\Sigma W^T)}^{1/k}$ and $\det{(\Sigma)}^{1/n}$ should be greater than the lower bound $\epsilon$. The constant for uncertainty loss is $\eta=0.3466$ and the corresponding $\epsilon=0.5$. This value selected guarantee that when the singular values of $W$ are all $1$, the determinism emergence $\Delta\mathcal{J}_2=0.2439>0$ as $W=(-0.0819,0.1432,-0.8421,0.5135)$ satisfies $\det(W\Sigma W^T)^\frac{1}{k}/\det(\Sigma)^\frac{1}{n}=0.614>\epsilon$. 

We can plot probability density graphs for the four dimensions of micro-state noise $\varepsilon_t$ and macro-state noise $\varepsilon_{M,t}$ in Fig.\ref{fig:random_walk}b, respectively. We can find that the variance of macro-state is smaller than that of micro-state, so the determinism emergence of macro-state dynamical systems is larger. Meanwhile, in Fig.\ref{fig:random_walk}c, we can see that under the condition that the singular values of $W$ are all 1, the higher $k$, the stronger the uncertainty of the system, and the smaller the degree of causal emergence.

Reducing the average noise of the system and increasing its determinism is an important condition for the occurrence of causal emergence. The main reason is that we can reduce the dimensions with weaker determinism. The premise for this phenomenon is that there are significant differences in noise between different dimensions. As shown in Fig.\ref{fig:random_walk}d, we randomly generate $\Sigma$ and $W$ and calculate the value of $\Delta\mathcal{J}$. We found that the larger the variance of the eigenvalues $\kappa_i$ of $\Sigma$, $i=1,\dots,n$, the more likely the occurrence of causal emergence.
\begin{figure}[htbp]
	\centering
 \begin{minipage}[c]{0.4\textwidth}
		\centering
		\includegraphics[width=1\textwidth]{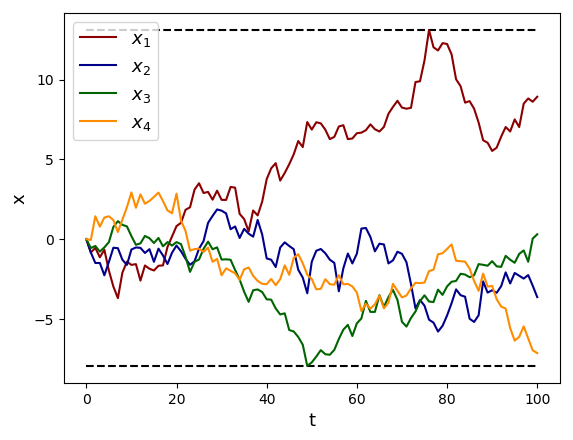}
		\centerline{(a)}
	\end{minipage}
	\begin{minipage}[c]{0.4\textwidth}
		\centering
		\includegraphics[width=1\textwidth]{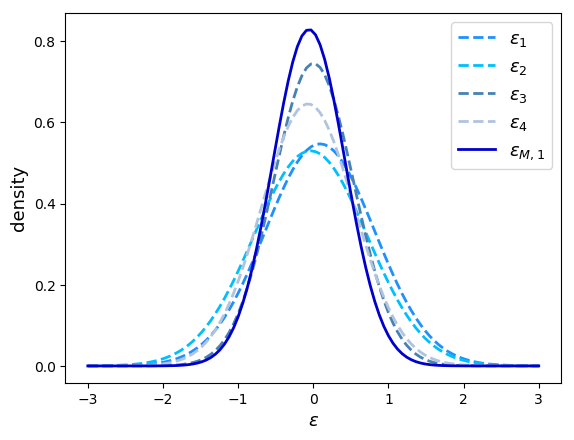}
		\centerline{(b)}
	\end{minipage}
	\begin{minipage}[c]{0.4\textwidth}
		\centering
		\includegraphics[width=1\textwidth]{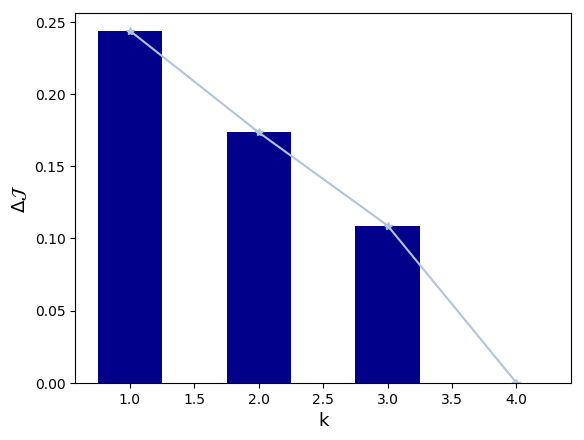}
		\centerline{(c)}
	\end{minipage}
        \begin{minipage}[c]{0.4\textwidth}
		\centering
		\includegraphics[width=1\textwidth]{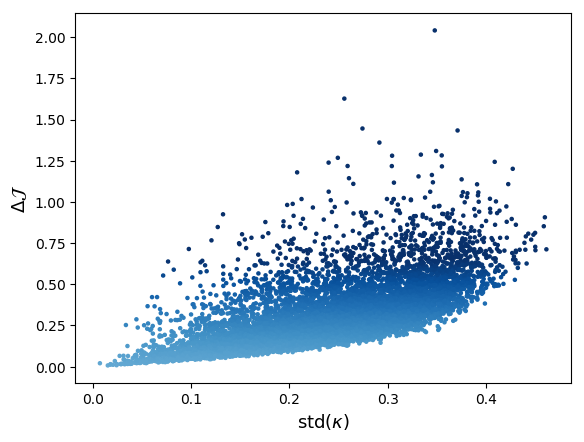}
		\centerline{(d)}
	\end{minipage}
	\caption{Experimental results for random walks model. (a) The trajectories $x_t$ of a random walker in different dimensions. (b) The probability density functions for the four dimensions of micro-state noises $\varepsilon_t$ and the macro-state noise $\varepsilon_{M,t}$. (c) The degree of causal emergence under different macro-state dimension $k$. Under the condition that the singular values of $W$ are all 1, the higher $k$, the stronger the uncertainty of the system, and the smaller the degree of causal emergence. (d) $\Delta\mathcal{J}$ and $std(\kappa)$. The larger the variance of the eigenvalues $\kappa_i$ of $\Sigma$, $i=1,\dots,n$, the more likely the occurrences of causal emergence with higher degrees.}
 \label{fig:random_walk}
\end{figure}

\subsection{Heat dissipation}\label{secEnergydissipation}
The first case focuses on determinism and noise, while the second case will focus on degeneracy and parameter matrix $A$. This case is called the discretized \cite{Peaceman1955} heat conduction model \cite{Bergman2011,Patankar2018}. In this discretized case, the conduction of heat is mainly reflected in the temperature changes over time \cite{LeVeque2007} at the corresponding observation nodes \cite{Roy2020}. We can use matrix $A$ to represent the changes in temperature, which includes information on the rate of change. This matrix usually corresponds to a positive definite sparse matrix, which can be used to describe a system's temperature change process. For example
\begin{eqnarray}\label{dissipation}
	x_{t+1}=Ax_t+\varepsilon_t, \varepsilon_t\sim \mathcal{N}(0,\sigma^2 I_n), \sigma=0.01,
\end{eqnarray}
describes the variation of temperature $x_t$ over time $t$. When $n=4$, $x_t=(x_{1,t},x_{2,t},x_{3,t},x_{4,t})^{T}$, $x_{1,t}$, $x_{2,t}$, $x_{3,t}$, and $x_{4,t}$ represent the temperatures at four nodes at time $t$ as shown in Fig.\ref{fig:energy}b. The parameter matrix is,
\begin{eqnarray}\label{Aenergy}
\begin{gathered}
	A=\begin{pmatrix}  0.6&0.2&0&0\\
	0.2&0.7&0.1&0\\
    0&0.1&0.4&0.1\\
    0&0&0.1&0.3\end{pmatrix},
\end{gathered}
\end{eqnarray}
where diagonal elements represent the rates of heat conservation at each node, and the off-diagonals represent the energy transfer between different nodes as shown in Fig.\ref{fig:energy}a. Assuming $\epsilon=1$ as $\eta=0$, then the uncertainties in both the dynamics at micro- and macro- levels are the same according to Equation (\ref{SetWSigEntropy}). As a result, only the effect of degeneracy emergence can influence the causal emergence according to Equation (\ref{maxCE12}).

We can simulate the diffusion process by setting the initial state to $x_0=(10,10,10,10)^{T}$ and we iterate the equation $x_{t+1}=Ax_t+\varepsilon_t$ for $t=1,2,\dots$. The results are shown in Fig\ref{fig:energy}b. We can directly coarse-grain the micro-state data to obtain $y_t=Wx_t$ for $t=0,1,2,\dots$. At the same time, we can use the macro-state dynamics to iterate the initial macro-state $y_0$. In order to distinguish from the previous $y_t$ data, we use $\hat{y}_t$ to represent the data generated by macro-state dynamical iterations here. $\hat{y}_0=y_0$ and $\hat{y}_{t+1}=WAW^\dagger\hat{y}_{t}+W\varepsilon_t$ for $t=0,1,2,\dots$. By comparison in Fig.\ref{fig:energy}c, we can find that the dynamics of $y_t$ and $\hat{y}_t$ are close to each other, macro-state dynamics can represent the dynamic evolution law of data after the coarse-graining.

By reducing the low correlation and dimensions that are difficult to reverse between $x_t$ and $x_{t+1}$, degeneracy can be increased, resulting in stronger causal emergence. For example, in the case shown in Fig.\ref{fig:energy}d, when $k=1$, $\sigma=0.01$ and $\eta=0$ as Shannon entropy of $x_t$ and $y_t$ remains stable, causal emergence occurs with the degree $\Delta\mathcal{J}=0.6656$ and it is greater than the case when $k>1$. We can take any element from $\mathcal{W}^*$, for example, $W=(0.5856, 0.7910 , 0.1748, 0.03065)\in \mathcal{R}^{1\times k}$ which satisfies our condition of $W$ in Equation (\ref{WSimas0}). The macro-state dynamic parameter matrix is $A_M\in \mathcal{R}^{1\times 1}$ with an eigenvalue $\lambda_{M,1} = 0.8702$. The eigenvalues of the micro-state matrix $A$ are $\lambda_1=0.8702, \lambda_2= 0.5000, \lambda_3= 0.4000$, and $\lambda_4=0.2298$. It can be seen that retaining the maximum eigenvalue will result in significant causal emergence. In physical terms, it can be understood as a region with slower dissipation, and the causal relationship between the temperature at time $t+1$ and time $t$ is stronger, making it more suitable for analyzing the temperature evolution law of the entire system. 

To verify our theoretic result about the exact expression for causal emergence, we can compare the obtained analytical solution in Equation (\ref{maxCE}) with the numerical results which can be obtained by the following steps. We first randomly generate numeric samples of $x_t\sim U([-1,1]^n)$ and $y_t\sim U([-1,1]^k)$, then perform a one-step iteration for micro and macro-dynamics as $x_{t+1}=Ax_t+\varepsilon_t$ and $y_{t+1}=A_My_t+\varepsilon_{M,t}$, $\varepsilon_t\sim \mathcal{N}_n(0,\sigma^2I_n)$ and $\varepsilon_{M,t}\sim \mathcal{N}_k(0,\sigma^2I_k)$. We then directly calculate the mutual information $I(x_{t+1},x_t|do(x_t)\sim U([-1,1]^n))$ for micro-dynamics and $I(y_{t+1},y_t|do(y_t)\sim U([-1,1]^k))$ between the generated samples of $x_{t+1}$ at time $t+1$ and the samples of $x_t$ at time $t$ to obtain the numerical solution of effective information and calculate causal emergence as
\begin{eqnarray}\label{generateI}
\Delta\mathcal{I}=\frac{I(y_{t+1},y_t|do(y_t)\sim U([-1,1]^k))}{k}-\frac{I(x_{t+1},x_t|do(x_t)\sim U([-1,1]^n))}{n}.
\end{eqnarray}
To distinguish from analytical solutions $\Delta\mathcal{J}$, we use $\Delta\mathcal{I}$ to represent the causal emergence computed by the mentioned numeric method. From Fig.\ref{fig:energy}e we found that when $n=4, k=1$ and $\sigma^2=0.01$, as the sample size increases, the numerical solutions for causal emergence gradually approach the analytical solutions.
\begin{figure}[htbp]
	\centering
 \begin{minipage}[c]{0.4\textwidth}
		\centering
		\includegraphics[width=1\textwidth]{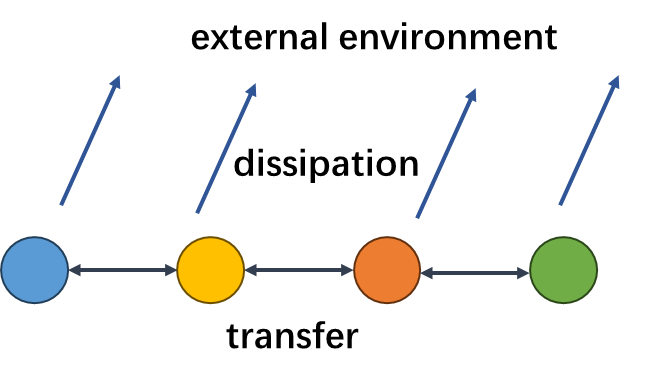}
		\centerline{(a)}
	\end{minipage}\\
\begin{minipage}[c]{0.4\textwidth}
		\centering
		\includegraphics[width=1\textwidth]{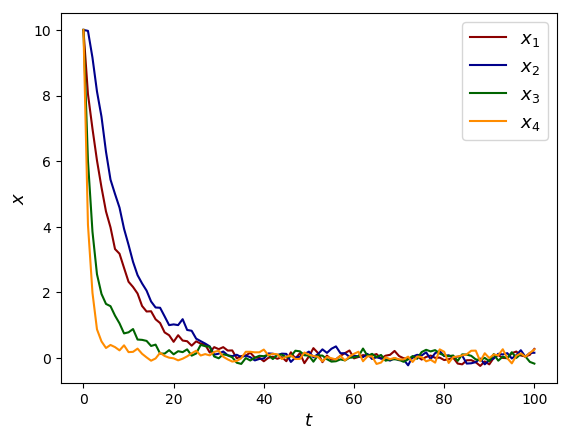}
		\centerline{(b)}
	\end{minipage}
	\begin{minipage}[c]{0.4\textwidth}
		\centering
		\includegraphics[width=1\textwidth]{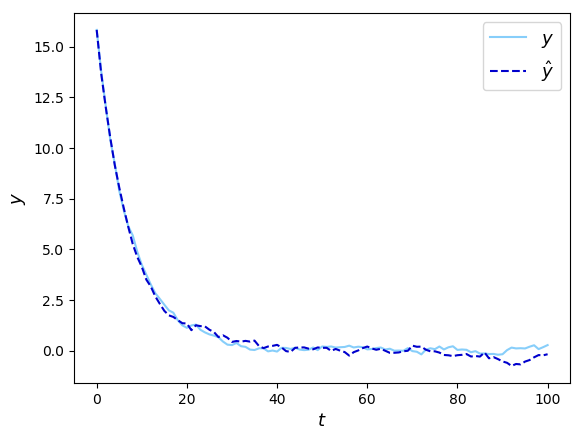}
		\centerline{(b)}
	\end{minipage}
	\begin{minipage}[c]{0.4\textwidth}
		\centering
		\includegraphics[width=1\textwidth]{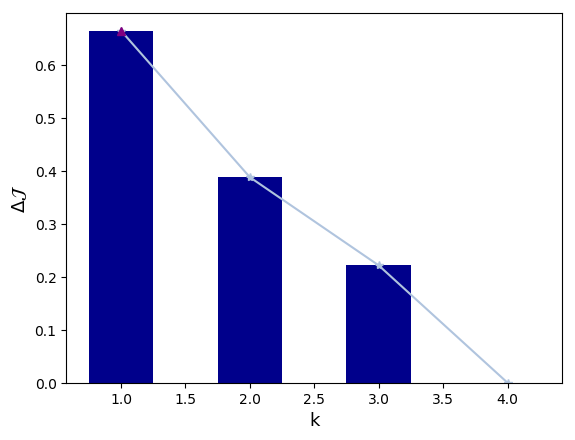}
		\centerline{(d)}
	\end{minipage}
        \begin{minipage}[c]{0.4\textwidth}
		\centering
		\includegraphics[width=1\textwidth]{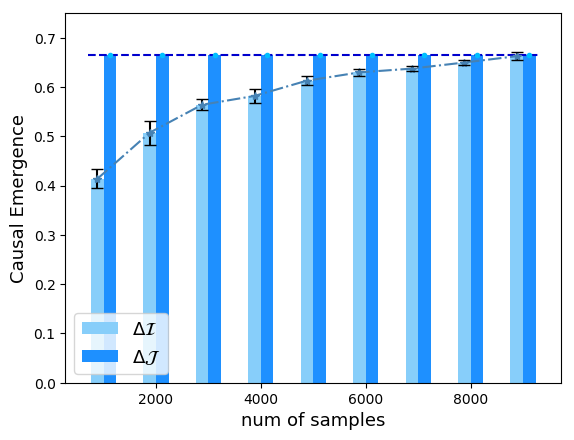}
		\centerline{(e)}
	\end{minipage}
	\caption{Experimental results for heat dissipation model. (a) The illustration of heat transfer and dissipation for a one-dimensional chain with four vertices. (b) The micro-state data of temperature at different times $t$ for 4 nodes. (c) The comparison between the macro-state data directly obtained by coarsening the micro-state data, $y_t$ to obtain $y_t=Wx_t$ for $t=0,1,2,\dots$, and $\hat{y}_t$ which is generated by the macro-state dynamics to iterate the initial macro-state $y_0$. (d) The comparisons for the degrees of causal emergence under different dimensions of macro-states $k$.  when $k=1$ and $\epsilon=1$, maximized degree of causal emergence occurs as $\Delta\mathcal{J}=0.6656$. (e) The comparison between $\Delta\mathcal{I}$ and $\Delta\mathcal{J}$. As the sample size increases, the numerical solutions for causal emergence $\Delta\mathcal{I}$ gradually approach the analytical solutions $\Delta\mathcal{J}$.}
    \label{fig:energy}
\end{figure}

\subsection{Spiral rotating}
\begin{figure}[htbp]
	\centering
\begin{minipage}[c]{0.4\textwidth}
		\centering
		\includegraphics[width=1\textwidth]{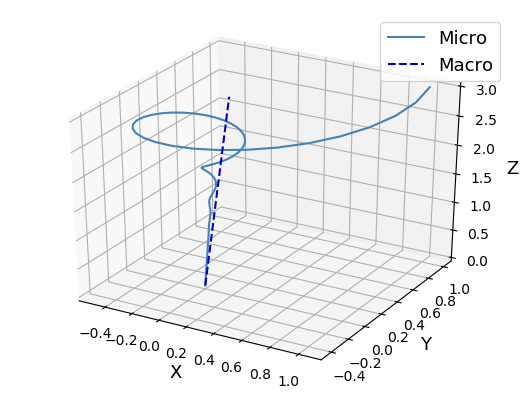}
		\centerline{(a)}
	\end{minipage}
	\begin{minipage}[c]{0.4\textwidth}
		\centering
		\includegraphics[width=1\textwidth]{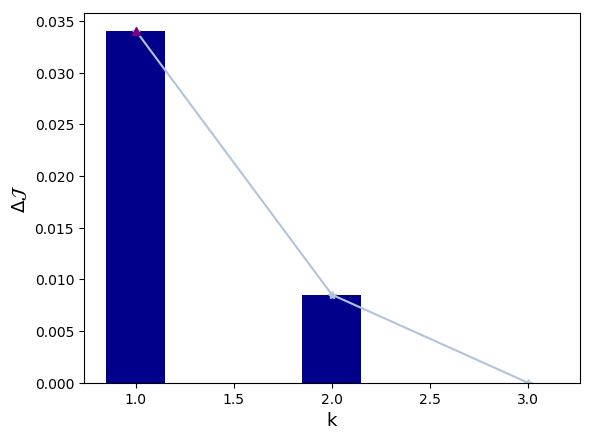}
		\centerline{(b)}
	\end{minipage}
	\begin{minipage}[c]{0.4\textwidth}
		\centering
		\includegraphics[width=1\textwidth]{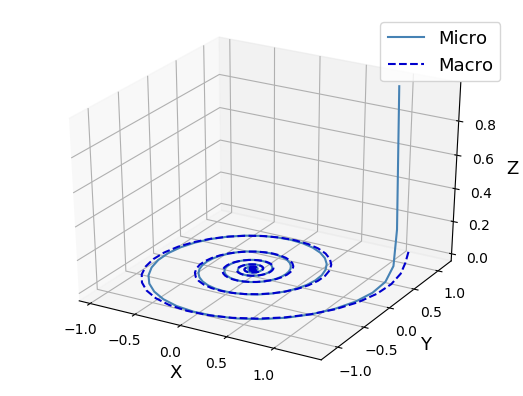}
		\centerline{(c)}
	\end{minipage}
        \begin{minipage}[c]{0.4\textwidth}
		\centering
		\includegraphics[width=1\textwidth]{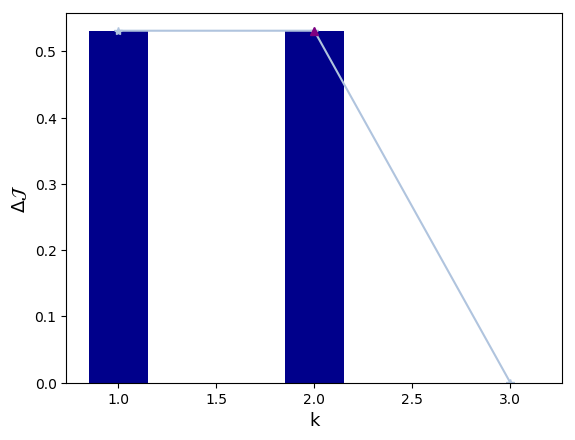}
		\centerline{(d)}
	\end{minipage}
 \begin{minipage}[c]{0.4\textwidth}
		\centering
		\includegraphics[width=1\textwidth]{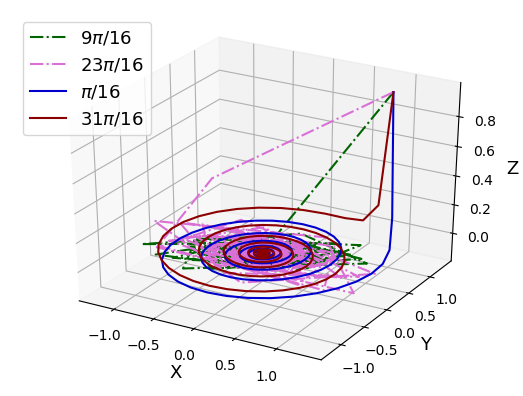}
		\centerline{(e)}
	\end{minipage}
        \begin{minipage}[c]{0.4\textwidth}
		\centering
		\includegraphics[width=1\textwidth]{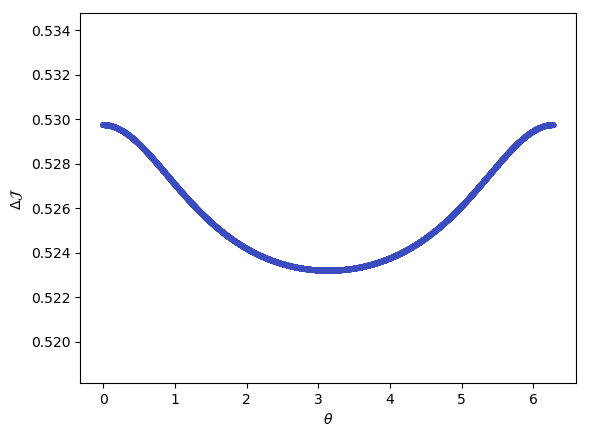}
		\centerline{(f)}
	\end{minipage}
	\caption{Experimental results of spiral rotating. (a) When $\Psi={\rm diag}(0.94,0.94,0.99)$ and $x_0=(1,1,3)^T$, $x_t $ represents a point in space that rotates around the axis of rotation and contracts towards the axis of rotation. By the optimal coarse-graining $W^*$, we can obtain macro-states that move along the axis of rotation. (b) The degree of causal emergence and its dependence on the dimension of macro-states for the example in (a), which reaches its maximum value as $\Delta\mathcal{J}=0.0341$ when $k=1$. (c) The spiral curve when $\Psi={\rm diag}(0.99,0.97,0.2)$ and $x_0=(1,1,1)^T$, the trajectory of $x_t$ is compressed to a plane perpendicular to $u$ at the initial stage. (d) We can project $x_t$ on the plane as a macro-state, where $k=2$, $\Delta\mathcal{J}=0.5295$. (e) When $\pi/2<\theta<3\pi/2$, $x_t$ tends to oscillate more than rotate, while $\theta<\pi/2$ or $\theta>3\pi/2$ $x_t$ tends to be a rotate model. (f) When $\theta$ approaches $\pi$, the degree of causal emergence is smaller, while when it approaches $0$ or $2\pi$, the degree of causal emergence is larger.}
 	\label{fig:rotation}
\end{figure}

In addition to the results obtained for specific models in sections \ref{secRandomWalk} and \ref{secEnergydissipation}, we need a more intuitive understanding of the meaning of coarse-graining and causal emergence. Here, we take a spiral rotating model \cite{Weisstein2003} from analytical geometry \cite{Qiu1996,Boyer2012} as an example to visualize our model system in $\mathcal{R}^3$ space and analyze how causal emergence is reflected. A common example is a model in which a rigid body rotates about an axis. Suppose we have a vector in three-dimensional space that represents the position of a point $x_0\in \mathcal{R}^3$, and we want to rotate it about some straight line through the origin. This operation can be represented by the rotation matrix $R$. For rotation around the straight line represented by the unit vector $u=(a,b,c)$, $||u||=\sqrt{a^2+b^2+c^2}=1$, a rotation matrix can be used to describe the rotation operation:
\begin{eqnarray}\label{R}
	\begin{gathered}
		R=\begin{pmatrix}  {\rm cos}\theta+a^2 (1-{\rm cos}\theta)&ab(1-{\rm cos}\theta)-c {\rm sin}\theta & ac(1-{\rm cos}\theta )+b {\rm sin}\theta \\ ab(1-{\rm cos}\theta )+c {\rm sin}\theta  &{\rm cos}\theta+b^2 (1-{\rm cos}\theta)&bc(1-{\rm cos}\theta )-a {\rm sin}\theta \\ac(1-{\rm cos}\theta )-b {\rm sin}\theta& bc(1-{\rm cos}\theta )+a {\rm sin}\theta&{\rm cos}\theta+c^2 (1-{\rm cos}\theta)\end{pmatrix}.
	\end{gathered}
\end{eqnarray}
Where $\theta$ represents the angle of the rotation. This matrix describes a rotation about the straight line represented by the unit vector $u$. Multiply the vector $x_t$ by the matrix $R$ to get the new rotated vector $x_{t+1}$. This rotation matrix describes a model of rotation around a certain straight line and can be used to describe rotation operations around a specified axis in three-dimensional space. Since $R$ doesn't change the modulus of $x_t$, we can add an adjustment matrix $\Psi={\rm diag}(\psi_1,\psi_2,\psi_3)$ to adjust the modulus of the variable $x_t$. Then the linear stochastic iteration systems in $\mathcal{R}^3$ is 
\begin{eqnarray}\label{rotation}
	x_{t+1}=Ax_t+\varepsilon_t, A=R\Psi,\varepsilon_t\sim \mathcal{N}(0,\sigma^2 I_n), \sigma=0.01,
\end{eqnarray}
when $\psi_i<1,i=1,2,3$, the model is attenuated. Adjusting $\psi_i$ can cause the system to produce different values of causal emergence at different macro-state dimensions. 

Assuming the direction vector of the rotation axis is $u_0=(0,0.1,1)^T$, after normalization as $u=u_0/||u_0||$, the rotation axis $u$ can be obtained. At the same time, we specify the rotation angle $\theta=\pi/16$ and $\eta=0$ as the dimension averaged Shannon entropy of both $x_t$ and $y_t$ remains the same. Through $\theta$ and $u$, we can get the rotation matrix $R$. By changing $\Psi$, we will obtain causal emergence in different forms. 

The occurrence of causal emergence in three-dimensional space can be understood as the stronger causal effect of system evolution when the path of $x_t$ is projected onto a plane or a straight line. As shown in Fig.\ref{fig:rotation}a, when $\Psi={\rm diag}(0.94,0.94,0.99)$ and $x_0=(1,1,3)^T$, $x_t$ contracts towards the axis of rotation. In Fig.\ref{fig:rotation}b, the causal emergence reaches its maximum value as $\Delta\mathcal{J}=0.0341$ when $k=1$. The eigenvalues of the micro-state matrix $A$ are $\lambda_1=0.9895$,$\lambda_2=0.9222+0.1834i$ and $\lambda_3=0.9222-0.1834i$, the macro-state parameter matrix $A_M$ only retains the eigenvalues with the maximum modulus $\lambda_1=0.9895$. This result is obtained by setting the coarse-graining mapping $W$ to compress the rotation trajectory into a one-dimensional straight line as the macro-states. 

In the second scenario in Fig.\ref{fig:rotation}c, when $\Psi={\rm diag}(0.99,0.97,0.2)$ and $x_0=(1,1,1)^T$, the path of $x_t$ will first shrink to a plane perpendicular to $u$. In \ref{fig:rotation}d, we project $x_t$ onto the plane as the macro-state, where $k=2$, $\Delta\mathcal{J}=0.5295$. We retain the two eigenvalues $\lambda_1=0.9612 +0.1900i$ and $\lambda_2=0.9612 -0.1900i$ with the largest and identical module lengths. The reason why we do not take $k=1$ here is that the more dimensionality we reduce, the greater the information loss and $|\lambda_1|=|\lambda_2|$. Therefore, in the absence of further improvement in causal emergence, we do not need to continue the dimensionality reduction of the system. 

Besides $\Psi$, rotation angle $\theta$ also significantly impacts causal emergence. In Fig.\ref{fig:rotation}f, when $\theta$ approaches $\pi$, the degree of causal emergence is smaller, while when it approaches $0$ or $2\pi$, it is larger. This is because when $\pi/2<\theta<3\pi/2$, $x_t$ tends to oscillate more than rotate, as shown in the Fig.\ref{fig:rotation}e. The correlation between $x_t+1$ and $x_t$ in the model is relatively small, leading to a weakening of causal emergence.

\section{Discussion}
\subsection{Optimal causal emergence when $W^{\dagger}=W^T$}

\begin{figure}[htbp]
	\centering
 \begin{minipage}[c]{0.4\textwidth}
		\centering
		\includegraphics[width=1\textwidth]{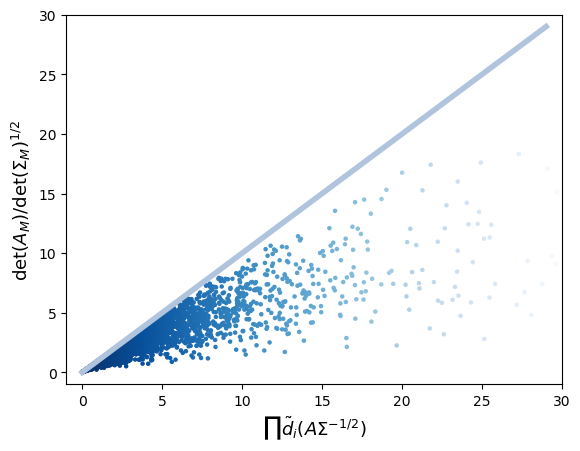}
		\centerline{(a)}
	\end{minipage}
 \begin{minipage}[c]{0.4\textwidth}
		\centering
		\includegraphics[width=1\textwidth]{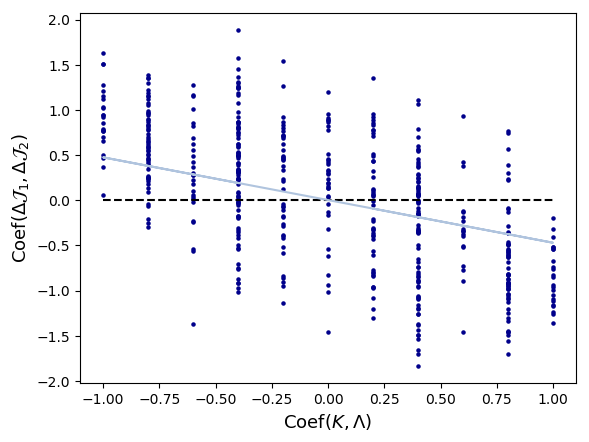}
		\centerline{(b)}
	\end{minipage}
	\caption{(a) We randomly generate matrices $A$, $\Sigma$ and $W$ with a singular value of 1 when $n=4$, $k=2$, we can see the upper bound of $|\det(WAW^\dagger)|/\det(W\Sigma W^T)^{1/2}$ is $\tilde{d}_1\tilde{d}_2$. (b) When $\lambda_1=0.8,\lambda_1=0.2,\lambda_3=0.4,\lambda_4=0.2$ and randomly arrange $(0.2, 0.4, 0.6, 0.8)$ to generate $(\kappa_1,\dots,\kappa_n)$. $C(\Delta\mathcal{J}_1,\Delta\mathcal{J}_2)$ and $C(\kappa,\lambda)$ shows a negative slope after drawing the scatter plot.}
	\label{fig:uW}
\end{figure}

In previous sections, we have obtained the explicit expression for the maximum causal emergence and the optimal coarse-graining under the condition that the dimension averaged uncertainty loss by coarse-graining is bounded. However, under this condition, the competition between the degeneracy term which is solely determined by $A$ and the determinism term which is determined by $A$ and $\Sigma$ in the optimal causal emergence can not be observed as shown in Equation (\ref{maxCE12}). 

Therefore, we will discuss the maximization of causal emergence under a different condition, that is,
\begin{equation}
    \label{eq.WdaggerT}
    W^{\dagger}=W^T,
\end{equation}
in this section. Equation (\ref{eq.WdaggerT}) means the coarse-graining map $W$ can be decomposed as a projection by discarding the information in $n-k$ dimensions and a rotation in a $k$-dimensional space without information loss. This also means the $k$ first singular values are units, and the norm of $W$ and $\Sigma$ will not be excessively enlarged or reduced. In this case, as $W^\dagger=W^T$, the coarse-graining strategy will simultaneously filter the eigenvalues of $A$ and $\Sigma$. According to Equation (\ref{EI}) and (\ref{Causal-emergence}), we can get that $|\det(A)|/\det(\Sigma)^{1/2}=|\det(A\Sigma^{-1/2})|$. So  $A\Sigma^{-1/2}$ determines the causal emergence of the system and is necessary to be analyzed below.  

\begin{thm}\label{JASigmaV}
When $W^{\dagger}=W^T$, the causal emergence of
linear stochastic iterative systems satisfies
   \begin{eqnarray}\label{Jdleq}
\Delta\mathcal{J}\leq\frac{1}{k}\sum_{i=1}^{k}\ln\displaystyle|\tilde{d}_i|-\frac{1}{n}\sum_{i=1}^{n}\ln\displaystyle|\tilde{d}_i|,
\end{eqnarray} 
where, $|\tilde{d}_1|\geq\dots\geq|\tilde{d}_n|$ are $n$ eigenvalues of $A\Sigma^{-1/2}$. The equal sign holds when $A$ and $\Sigma$ share the same $n$ eigenvectors as
\begin{equation}
\label{eq.sameeigens}
    A\Sigma^{-1/2}=V\Lambda K^{-1/2}V^T=V\Tilde{D}V^T,
\end{equation}where $\Tilde{D}=\Lambda K^{-1/2}={\rm diag}(\tilde{d}_1,\dots,\tilde{d}_n)$ is also a diagonal matrix, in which $|\tilde{d}_1|\geq\dots\geq|\tilde{d}_n|$, and $V$ is the matrix formed by the shared eigenvectors of $A$ and $\Sigma$.
\end{thm}

When $A$ and $\Sigma$ do not share the eigenvector matrix, we can't directly calculate the correlation between the eigenvalue matrix of $A\Sigma^{-1/2}$ and $K$ or $\Lambda$ within the current framework of matrix theory, so it is difficult to find an analytical solution for causal emergence. But we can still obtain numerical solutions through optimization of $W$ and confirm that the upper bound of causal emergence is to retain the maximum $k$ eigenvalues of $A\Sigma^{-1/2}$. We have the inequality (\ref{Jdleq}) in Theorem \ref{JASigmaV}. To verify this inequality we randomly generate matrices $A$, $\Sigma$ and $W$ with a singular value of 1 when $n=4$, $k=2$, we can see the upper bound of $|\det(WAW^\dagger)|/\det(W\Sigma W^T)^{1/2}$ is $\tilde{d}_1\tilde{d}_2$ in Fig.\ref{fig:uW}b. Unlimited the singular value of $W$ means that it can reduce the macro-state noise and information entropy in the system. This also means that the locally effective information we retain can be adjusted arbitrarily by adjusting $W$. This is also a balance point between determinism and degeneracy. 


Then, we can derive an analytical expression for the maximized causal emergence when $A$ and $\Sigma$ share the same $n$ eigenvectors, that is to say, $\Sigma=VKV^T$ and $A=V\Lambda V^T$, and Equation (\ref{eq.sameeigens}) holds. Rearranging the order of $\tilde{d_i}$ can result in diagonal matrix $\tilde{D}={\rm diag}(\tilde{d}_1,\dots,\tilde{d}_n)$ in which $|\tilde{d}_1|\geq\dots\geq|\tilde{d}_n|$. Then, according to Equation (\ref{Causal-emergence}) we can obtain an analytical solution for the maximized causal emergence, which is also the case where Equation (\ref{Jdleq}) takes the equal sign. The above proof process can refer to \ref{appendixAM}.



Next, we will study when the relationship between the determinism($\Delta \mathcal{J}_2$) and the degeneracy($\Delta \mathcal{J}_1$) is cooperative such that the causal emergence can be maximized under the condition of the simultaneous diagonalization for $A$ and $\Sigma$. This problem is equivalent to ask when the quantity expressed in Equation \ref{Jdleq} which is determined by the eigenvalues of $A\Sigma^{-1/2}$ will be maximized if the eigenvalues in $A$ and $\Sigma$ are unchanged. We found that the orders of the eigenvalues in both $A$ and $\Sigma$, which can be characterized by the correlation between $\kappa=(\kappa_1,\dots,\kappa_n)$ and $\lambda=(\lambda_1,\dots,\lambda_n)$ denoted as $C(\kappa,\lambda)$, play an important role. If $\lambda$ and $\kappa$ are positive correlated, the effects of determinism and degeneracy contributed to the causal emergence cancel each other. For example, when $\lambda=(0.8,0.6,0.4,0.2)$, $\kappa=(0.8,0.6,0.4,0.2)$, and $k=2$, the optimal coarse-graining strategy for $\Delta\mathcal{J}_1$ should retain the largest $k=2$ eigenvalues $\lambda_1=0.8,\lambda_2=0.6$ of $A$ according to Equation (\ref{maxDegeneracy}), but for $\kappa$, the largest $k$ eigenvalues $\kappa_1=0.8,\kappa_2=0.6$ are also retained, which implies small determinism and large degeneracy. We can see that the determinism and the degeneracy are mutually inhibiting leading to a small causal emergence: $\Delta \mathcal{J}=0$. 

On the other hand, if $\lambda$ and $\kappa$ are negative correlated then causal emergence can be increased because determinism and degeneracy reinforce each other. For example, when $\lambda=(0.8,0.6,0.4,0.2)$ and $\kappa=(0.2,0.4,0.6,0.8)$, the optimal coarse-graining strategy for $\Delta\mathcal{J}_1$ should retain the largest $k=2$ eigenvalues $\lambda_1=0.8,\lambda_2=0.6$ of $A$, this also means to retain the smallest $k$ eigenvalues of $\Sigma$, then determinism and degeneracy are mutually reinforced, and a larger causal emergence can be obtained: $\Delta \mathcal{J}=0.8959$. 

In Fig.\ref{fig:uW}a, we visualize how the correlation between $\lambda$ and $\kappa$ can influence the cooperative relationship between determinism and degeneracy which can be characterized by the correlation between $\Delta\mathcal{J}_1$ and $\Delta\mathcal{J}_2$ denoted as $C(\Delta\mathcal{J}_1,\Delta\mathcal{J}_2)$ in a small example. First, we fix the values and the orders of $\Lambda$ as $\lambda_1=0.8,\lambda_1=0.2,\lambda_3=0.4,\lambda_4=0.2$. Second, we randomly permute $(0.2, 0.4, 0.6, 0.8)$ to generate $\kappa$. Finally, we randomly sample the invertible matrix $V$ such that $A$ and $\Sigma$ can be obtained. The results are shown in Fig.\ref{fig:uW}a, $C(\Delta\mathcal{J}_1,\Delta\mathcal{J}_2)$ and $C(\kappa,\lambda)$ shows a negative slope after drawing the scatter plot, this means that the two values are negatively correlated. In general, maximizing causal emergence $\Delta\mathcal{J}$ is also a process of finding optimal orders to arrange $\lambda_i$ and $\kappa_i$ for all $i$ to achieve a balance between determinism $\Delta\mathcal{J}_2$ and degeneracy $\Delta\mathcal{J}_1$.

\subsection{About the information loss in dynamics prediction}
Our entire paper discusses the changes in causal emergence of systems after coarse-graining. For linear stochastic iteration systems, deleting smaller eigenvalues in the parameter matrix $A$ or larger eigenvalues in $\Sigma$ can both increase the value of causal emergence. Coarse graining and data dimensionality reduction \cite{Villemagne1987,Boley1994,Gallivan1994,Antoulas2005,Gugercin2008} are intricately linked. Once data dimensionality is reduced, it will be accompanied by information loss, and the transformation from micro-states to macro-states is no exception. Therefore, we also need to briefly analyze the relationship between the loss function in dynamics prediction and the effective information of the dynamics. For linear stochastic iteration systems, we define the dynamical loss in the following way:
\begin{definition}
	\label{dfn.dynamical-loss}
	(Dynamical loss):
	For the variable $x_t$ at time $t$, we define the dynamical loss as
	\begin{eqnarray}\label{LD}
		L_D=||x_{t+1}-\tilde{x}_{t+1}||,
	\end{eqnarray}
    in which $x_{t+1}=Ax_t+\varepsilon_t$, $\tilde{x}_{t+1}=\phi^\dagger(y_{t+1})=W^\dagger y_{t+1}$. 
\end{definition}
In Definition \ref{dfn.dynamical-loss}, $y_{t+1}$ is derived from $y_t$ and the macro iterative Equation (\ref{MacroDynamics}) and $y_t=Wx_t$. Therefore,
\begin{eqnarray}\label{tildex}
    \tilde{x}_{t+1}=W^\dagger WAW^\dagger Wx_t+W^\dagger W\varepsilon_t,
\end{eqnarray}
the dynamical loss then should be:
\begin{eqnarray}\label{dynamicalloss}
    L_D=||(W^\dagger WAW^\dagger W-A)x_t+(W^\dagger W-I)\varepsilon_t||,
\end{eqnarray}
it is the difference between the micro-state variables $x_t\in \mathcal{R}^n$ that are then mapped into micro-states after one iteration in the macro dynamics, and the actual micro-state $x_{t+1}$, which reflects the information loss caused by coarse-graining on dynamics. 

Due to the dependence between $L_D$ and variables $x_t$ and $\varepsilon_t$, in order to directly analyze the relationship between the causal emergence of the system $\Delta\mathcal{J}$ and the dynamical loss, we consider the upper bound of $L_D$ to eliminate the influence of random variables.
\begin{lem}
	\label{thm.supdynamical-loss}
	(Supremum of dynamical loss):
	For the variable $x_t$ at time $t$, when $||x_t||^*=\mathop{\sup}_{t}||x_t||<\infty$, we define the supremum of dynamical loss $L_D$ as
	\begin{eqnarray}\label{SD}
		L_D\leq S_D=\mathop{||A-\hat{A}||_F||x_t||^*}_{Degeneracy Supremum}+\mathop{(n-k)||\varepsilon_t||}_{Determinism Supremum},
	\end{eqnarray}
	in which $\hat{A}=W^\dagger WAW^\dagger W$. $S_D$ can be decomposed into determinism supremum and degeneracy supremum. $||\cdot||_F$ is the Frobenius norm of matrix $\cdot$.
\end{lem}
We can directly determine the relationship between the upper bound and $\Delta\mathcal{J}_1$. When $k$ is a constant, the minimum value of $S_D$ is equivalent to the maximum value of degeneracy emergence, $\Delta\mathcal{J}_1$.
\begin{thm}
	\label{thm.CESD}
	($\Delta\mathcal{J}_1$ and $S_D$):
	For the known variables $x_t$, the micro-state parameter matrix $A$ and the noise $\varepsilon_t$, $t=0,1,...$,
	\begin{eqnarray}\label{CESD}
		\arg\max\limits_{W}\Delta\mathcal{J}_1=\arg\min\limits_{W}S_D=W^*.
	\end{eqnarray}
\end{thm}
The proof of this theorem is referred to Appendix \ref{LossandJ}. After obtaining the relationship between the supremum $S_D$ and the degeneracy emergence $\Delta\mathcal{J}_1$. It can be observed that after minimizing $S_D$, the solution set of $W$, $\mathcal{W}_L$, satisfies $\mathcal{W}_L\in\mathcal{W}^*$. It can be seen that reducing loss can only increase the degeneracy of the system, but further search for $W$ is needed for determinism. This explains why we still need to train for a few steps to find the maximum causal emergence after the loss function converges in machine learning for causal emergence identification \cite{Zhang2022,Yang2023}.
\subsection{Nonlinear form}
Although the effective information and causal emergence in this paper are based on linear stochastic iterative systems, we can still apply them to nonlinear models under certain conditions. Nonlinear stochastic iterative systems like $x_{t+1}=f(x_t)+\varepsilon_t$ do not have the same known parameter matrix $A$ as linear stochastic iterative systems. However, when $f:\mathcal{R}^n\to \mathcal{R}^n$, we can replace $A$ with the gradient matrix of $\nabla f(x_t)\in \mathcal{R}^{n\times n}$ at $x_t$ 

Therefore, near $x_t$, when the noise is small, we can obtain an expression for the effective information of the nonlinear iterative system
\begin{eqnarray}\label{fxhJ}
    \mathcal{J}(x_t)=\ln\displaystyle\frac{|\det(\nabla f(x_t))|^\frac{1}{n}L}{(2\pi e)^\frac{1}{2}\displaystyle \det(\Sigma)^\frac{1}{2n}},
\end{eqnarray}

In this way, we can get the expression of causal emergence as
\begin{eqnarray}\label{fxCE}
\Delta\mathcal{J}(x_t)=\ln\frac{|\det(\nabla f_M(y_t)|^\frac{1}{k}}{|\det(\nabla f(x_t)|^\frac{1}{n}}+\ln\frac{|\det(\Sigma)|^\frac{1}{2n}}{|\det(\Sigma_M)|^\frac{1}{2k}}.
\end{eqnarray}
For some simple differentiable nonlinear functions $f(x_t)$, we can use this method to calculate the magnitude of causal emergence.

\section{Conclusion}
In this article, we first propose the coarse-graining strategy $y_t=Wx_t$ of linear stochastic iterative systems like $x_{t+1}=Ax_t+\varepsilon_t, \varepsilon_t\sim \mathcal{N}(0,\Sigma)$, that map micro-state $x_t\in \mathcal{R}^n$ into macro-states $y_t\in \mathcal{R}^k, k<n$. Coarse-graining can naturally derive macro dynamical models $y_{t+1}=A_My_t+\varepsilon_{M,t}$ for $y_t\in \mathcal{R}^k$, $A_M=WAW^\dagger$, $\varepsilon_{M,t}\sim \mathcal{N}(0,\Sigma_M)$, $\Sigma_M=W\Sigma W^T$. For $x_{t+1}=Ax_t+\varepsilon_t$, we can use $\mathcal{J}$ and $\Delta\mathcal{J}$ to represent the effective information of the system and causal emergence in different dimensions $k$ and $n$. To ensure that the macro dynamical information entropy $H(p(y_t))$ is not excessively reduced during coarse-graining, we need to ensure 
$\frac{H(x_t)}{n}-\frac{H(y_t)}{k}$ is bounded by $\eta$.

To directly determine the causal emergence in the system without being affected by coarse-graining parameter matrix $W$, we optimize $W$ to obtain the maximum causal emergence $\Delta\mathcal{J}^*$ that the system can achieve. When obtaining the optimal solution, the iterative parameter matrix $A_M$ only retains the maximum $k$ eigenvalues of $A$, and the determinant value of $\Sigma_M$ is compressed to the lowest allowed lower limit $\epsilon\det(\Sigma)^\frac{k}{n}$. The obtained solution set can be understood as the intersection of $k$ hyperplanes and hyperellipsoids, an ellipse in three-dimensional space. We search for causal emergence in three systems that are known models, focusing on the application of determinism, degeneracy and the manifestation of causal emergence in $\mathcal{R}^3$ space. In random walks, energy dissipation, and accurate selection models, varying degrees of causal emergence can be detected.

We conclude that there are two influencing factors for causal emergence, one is the iterative parameter matrix $A$, and the other is the covariance matrix $\Sigma$ of random noise, which respectively affect the degeneracy and determinism of the system. Smaller eigenvalues in the parameter matrix indicate weak causal relationships between variables in certain dimensions. Suppose the coarse-grained strategy precisely discards these dimensions. In that case, the causal effect of the macro-state dynamics is stronger than that of the micro-state, and the optimal causal emergence will occur.

The causal emergence calculation method in this article has several advantages. Firstly, we present the causal emergence of a known model in a linear stochastic iterative system as an analytical solution. This is also the first time that causal emergence has been calculated as an analytical solution in a continuous-state dynamic system. Secondly, in maximizing causal emergence, we found the solution set of the coarse-graining strategy. We found that the eigenvalues of the parameter matrix in a linear system determine the value of causal emergence. The third is the final result we obtain, which is only related to the system itself and not to specific data. This to some extent solves the problem of numerical solutions relying on data and the need for artificial search for coarse-graining strategies.

However, our model still has many shortcomings. One is that the current model is only applicable to linear systems, and there is still a lack of unified analysis and judgment for nonlinear systems. On the other hand, the time of the linear stochastic iterative system we study is discrete, and the calculation method of causal emergence is not applicable to continuous time dynamic systems. 

There are two important research directions in the future. One is to consider how to correctly calculate causal emergence for continuous-time Markov processes with Discrete probability space such as birth and death processes or queuing processes. The other direction is about the causal emergence of stochastic differential equations such as option pricing equation and Langevin equation. We will continue to conduct in-depth research on these contents in the future.
\\

\textbf{Data ailability}

All the codes and data are available at: \url{https://github.com/kilovoltage/An_Exact_Causal_Emergence_Theory}.

\appendix
\renewcommand{\thesection}{Appendix \Alph{section}}
\renewcommand\theequation{\Alph{section}.\arabic{equation}} 
\section[\appendixname~\thesection]{Calculation effective information and causal emergence}\label{appendixEIandce}
In this appendix, we will prove the theorems related to effective information and causal emergence, and demonstrate the specific process of formula derivation. The first process to be demonstrated is the inference of effective information related to linear stochastic iterative systems $\mathcal{J}$.\\

\textbf{Proof of Theorem \ref{thm.Effective-information}}
\begin{proof}
	For the linear stochastic iteration system like Equation (\ref{MicroDynamics}), $x_{t+1}$ follows a normal distribution about $x_{t}$, so $x_{t+1}\sim \mathcal{N}(Ax_t,\Sigma)$,
	\begin{eqnarray}\label{p(xtp1|doxt)}
		p(x_{t+1}|x_t)=\displaystyle\frac{1}{\det(\Sigma)^\frac{1}{2}(2\pi)^\frac{n}{2}}\exp\left\{-\frac{1}{2}(x_{t+1}-Ax_t)^{'}\Sigma^{-1}(x_{t+1}-Ax_t)\right\}
	\end{eqnarray}
	Following the calculation method of causal geometry, assuming $\varepsilon_t$ is small, $do(x_t\sim U([-L/2,L/2]^n))\equiv do(x_t\sim U)$, then $p(x_t)=1/L^n, L>0$, we can calculate the effect distribution
	\begin{eqnarray}\label{ED}
		\begin{aligned}
			E_D(x_{t+1})&=\displaystyle\int_{\mathcal{X}} p(x_{t+1}|x_t)p(x_t)d^n x_t\\
			&=\displaystyle\int_{\mathcal{X}}  \displaystyle\frac{1}{\det(\Sigma)^\frac{1}{2}(2\pi)^\frac{n}{2}}\exp\left\{-\frac{1}{2}(x_{t+1}-Ax_t)^{'}\Sigma^{-1}(x_{t+1}-Ax_t)\right\}\frac{1}{L^n}d^nx_t\\
			_{(x_{t}\approx A^{-1}x_{t+1})}&\approx\displaystyle\int_{\mathcal{X}}\displaystyle\frac{1}{\det(\Sigma)^\frac{1}{2}(2\pi)^\frac{n}{2}}\exp\left\{-\frac{1}{2}(x_{t+1}-Ax_t)^{'}\Sigma^{-1}(x_{t+1}-Ax_t)\right\}\frac{|\det(A^{-1})|}{L^n}d^nx_{t+1}\\
			&=\frac{1}{|\det(A)|L^n}
		\end{aligned}
	\end{eqnarray}
	Afterward, we can calculate the effective information as
	\begin{eqnarray}\label{calEI}
		\begin{aligned}
			EI&=\displaystyle\int_{\mathcal{X}} D_{KL}[p(x_{t+1}|x_t)|E_D(x_{t+1})]d^n x_{t+1}\\
			&=\displaystyle\int_{\mathcal{X}}\frac{d^n x_{t}}{L^n}\displaystyle\int_{\mathcal{X}}  \displaystyle\frac{1}{\det(\Sigma)^\frac{1}{2}(2\pi)^\frac{n}{2}}\exp\left\{-\frac{1}{2}(x_{t+1}-Ax_t)^{'}\Sigma^{-1}(x_{t+1}-Ax_t)\right\}\\&\left[-\ln\left(\det(\Sigma)^\frac{1}{2}(2\pi)^\frac{n}{2}\right)-\frac{1}{2}(x_{t+1}-Ax_t)^{'}\Sigma^{-1}(x_{t+1}-Ax_t)-\ln\frac{1}{|\det(A)|L^n}\right]d^nx_{t+1}\\
			&=-\frac{1}{2}\ln\left(\det(\Sigma)(2\pi)^n\right)-\frac{n}{2}-\ln\frac{1}{|\det(A)|L^n}\\
			&=\ln\displaystyle\frac{|\det(A)|L^n}{\displaystyle \det(\Sigma)^\frac{1}{2}(2\pi e)^\frac{n}{2}}.
		\end{aligned}
	\end{eqnarray}
	According to the properties of information entropy, determinism information
	\begin{eqnarray}\label{Determinism}
		\begin{aligned}
			-\left<H(p(x_{t+1}|x_t))\right>&=-\left<-\int_{\mathcal{X}}p(x_{t+1}|x_t)\ln(p(x_{t+1}|x_t))d^nx_{t+1}\right>\\
			&=\displaystyle\int_{\mathcal{X}}\frac{d^n x_{t}}{L^n}\displaystyle\int_{\mathcal{X}}  \displaystyle\frac{1}{\det(\Sigma)^\frac{1}{2}(2\pi)^\frac{n}{2}}\exp\left\{-\frac{1}{2}(x_{t+1}-Ax_t)^{'}\Sigma^{-1}(x_{t+1}-Ax_t)\right\}\\&\left[-\ln\left(\det(\Sigma)^\frac{1}{2}(2\pi)^\frac{n}{2}\right)-\frac{1}{2}(x_{t+1}-Ax_t)^{'}\Sigma^{-1}(x_{t+1}-Ax_t)\right]d^nx_{t+1}\\
			&=\left[-\ln\left(\det(\Sigma)^\frac{1}{2}(2\pi)^\frac{n}{2}\right)-\frac{n}{2}\right]\\
			&=\ln\frac{1}{\det(\Sigma)^\frac{1}{2}(2\pi e)^\frac{n}{2}},
		\end{aligned}
	\end{eqnarray}
	degeneracy information
	\begin{eqnarray}\label{Degeneracy}
		\begin{aligned}
			H(E_D(x_{t+1}))&=-\int_{\mathcal{X}}E_D(x_{t+1})\ln(E_D(x_{t+1}))d^nx_{t+1}\\
			_{(x_{t}\approx A^{-1}x_{t+1})}&\approx-\displaystyle\int_{\mathcal{X}}E_D(x_{t+1})\ln\frac{1}{|det(A)|L^n}d^nx_{t+1}\\
			&=\ln\left(|det(A)|L^n\right).
		\end{aligned}
	\end{eqnarray}
	It's obvious that
	\begin{eqnarray}\label{DegeneracyDeterminismEI}
		-\left<H(p(x_{t+1}|x_t))\right>+H(E_D(x_{t+1}))=\ln\displaystyle\frac{|\det(A)|L^n}{\displaystyle \det(\Sigma)^\frac{1}{2}(2\pi e)^\frac{n}{2}}=EI.
	\end{eqnarray}
Then we can easily get
 \begin{eqnarray}\label{EI_appendix}
		\mathcal{J}\equiv \frac{EI}{n}=\frac{1}{n}\ln\displaystyle\frac{|\det(A)|L^n}{(2\pi e)^\frac{n}{2}\displaystyle \det(\Sigma)^\frac{1}{2}}=\ln\displaystyle\frac{|\det(A)|^\frac{1}{n}L}{(2\pi e)^\frac{1}{2}\displaystyle \det(\Sigma)^\frac{1}{2n}}.
	\end{eqnarray}
If the system is nonlinear, it can be understood that $A=\nabla f(x_t)$ will change with $x_t$. Then Equations (\ref{fxhJ}) and (\ref{fxCE}) can also be proven. 
\end{proof}

After obtaining effective information, we can calculate causal emergence.\\

\textbf{Non full rank matrix $A$}
\begin{proof}
	For the linear stochastic iteration system like Equation (\ref{MicroDynamics}), $y\in \mathcal{R}^m$ follows a normal distribution about $x\in \mathcal{R}^n$, so $y\sim \mathcal{N}(Ax,\Sigma)$, $A\in\mathcal{R}^{m\times n}$, $\Sigma\in\mathcal{R}^{m\times m}$,
	\begin{eqnarray}\label{p(xtp1|doxt)}
		p(y|x)=\displaystyle\frac{1}{\det(\Sigma)^\frac{1}{2}(2\pi)^\frac{m}{2}}\exp\left\{-\frac{1}{2}(y-Ax)^{'}\Sigma^{-1}(y-Ax)\right\}
	\end{eqnarray}
	Following the calculation method of causal geometry, assuming $\varepsilon_t$ is small, $do(x_t\sim U([-L/2,L/2]^n))\equiv do(x_t\sim U)$, then $p(x_t)=1/L^n, L>0$, we can calculate the effect distribution
	\begin{eqnarray}\label{ED}
		\begin{aligned}
			E_D(y)&=\displaystyle\int_{\mathcal{X}} p(y|x)p(x)d^n x\\
			&=\displaystyle\int_{\mathcal{X}}  \displaystyle\frac{1}{\det(\Sigma)^\frac{1}{2}(2\pi)^\frac{m}{2}}\exp\left\{-\frac{1}{2}(y-Ax)^{'}\Sigma^{-1}(y-Ax)\right\}\frac{1}{L^n}d^nx\\
        &=\displaystyle\int_{\mathcal{X}}\displaystyle\frac{1}{\det(\Sigma)^\frac{1}{2}(2\pi)^\frac{m}{2}}\exp\left\{-\frac{1}{2}(x-A^\dagger y)^{'}A^{'}\Sigma^{-1}A(x-A^\dagger y)\right\}\frac{1}{L^n}d^nx\\
			&=\frac{(2\pi)^\frac{n-m}{2}}{|{\rm pdet}(A^{'}\Sigma^{-1}A)|^\frac{1}{2}\det(\Sigma)^\frac{1}{2}}\displaystyle\int_{\mathcal{X}}\displaystyle\frac{1}{|{\rm pdet}(A^{'}\Sigma^{-1} A)|^{-\frac{1}{2}}(2\pi)^\frac{n}{2}}\exp\left\{-\frac{1}{2}(x-A^\dagger y)^{'}A^{'}\Sigma^{-1}A(x-A^\dagger y)\right\}\frac{1}{L^n}d^nx\\
   &\approx \frac{(2\pi)^\frac{n-m}{2}}{|{\rm pdet}(A^{'}\Sigma^{-1}A)|^\frac{1}{2}\det(\Sigma)^\frac{1}{2}L^n}
		\end{aligned}
	\end{eqnarray}
	According to the properties of information entropy, determinism information
	\begin{eqnarray}\label{Determinism}
		\begin{aligned}
			-\left<H(p(y|x))\right>&=-\left<-\int_{\mathcal{X}}p(y|x)\ln(p(y|x))d^my\right>\\
			&=\displaystyle\int_{\mathcal{X}}\frac{d^n x}{L^n}\displaystyle\int_{\mathcal{X}}  \displaystyle\frac{1}{\det(\Sigma)^\frac{1}{2}(2\pi)^\frac{m}{2}}\exp\left\{-\frac{1}{2}(x_{t+1}-Ax_t)^{'}\Sigma^{-1}(x_{t+1}-Ax_t)\right\}\\&\left[-\ln\left(\det(\Sigma)^\frac{1}{2}(2\pi)^\frac{m}{2}\right)-\frac{1}{2}(y-Ax)^{'}\Sigma^{-1}(y-Ax)\right]d^my\\
			&=\left[-\ln\left(\det(\Sigma)^\frac{1}{2}(2\pi)^\frac{m}{2}\right)-\frac{m}{2}\right]\\
			&=\ln\frac{1}{\det(\Sigma)^\frac{1}{2}(2\pi e)^\frac{m}{2}},
		\end{aligned}
	\end{eqnarray}
	degeneracy information
	\begin{eqnarray}\label{Degeneracy}
		\begin{aligned}
			-H(E_D(y))&=\int_{\mathcal{X}}E_D(y)\ln(E_D(y))d^my\\
			&=\displaystyle\int_{\mathcal{X}}E_D(y)\ln\left(\frac{(2\pi)^\frac{n-m}{2}}{|{\rm pdet}(A^{'}\Sigma^{-1}A)|^\frac{1}{2}\det(\Sigma)^\frac{1}{2}L^n}\right)\\
			&\approx\ln\left(\frac{(2\pi)^\frac{n-m}{2}}{|{\rm pdet}(A^{'}\Sigma^{-1}A)|^\frac{1}{2}\det(\Sigma)^\frac{1}{2}L^n}\right).
		\end{aligned}
	\end{eqnarray}
	It's obvious that
	\begin{eqnarray}\label{DegeneracyDeterminismEI}
		EI=-\left<H(p(x_{t+1}|x_t))\right>-(-H(E_D(x_{t+1})))=\ln\left(\frac{|{\rm pdet}(A^{'}\Sigma^{-1}A)|^\frac{1}{2}L^n}{(2\pi)^\frac{n}{2}e^\frac{m}{2}}\right).
	\end{eqnarray}
\end{proof}

After obtaining effective information, we can calculate causal emergence.\\

\textbf{Add observation noise}
\begin{proof}
	For the linear stochastic iteration system like Equation (\ref{MicroDynamics}), $y\in \mathcal{R}^m$ follows a normal distribution about $x\in \mathcal{R}^n$, so $y\sim \mathcal{N}(Ax,\Sigma)$, $A\in\mathcal{R}^{m\times n}$, $\Sigma\in\mathcal{R}^{m\times m}$. We can add observation noise as
    \begin{eqnarray}\label{p(xtp1|doxt)}
		\mathbf{x}=x+\theta_x, \theta_x\sim\mathcal{N}(0,\Theta_x)\\
        \mathbf{y}=y+\theta_y, \theta_x\sim\mathcal{N}(0,\Theta_y)
	\end{eqnarray}
   Since $y=Ax+\varepsilon$, $\mathbf{y}=A(\mathbf{x}-\theta_x)+\varepsilon+\theta_y\sim\mathcal{N}(A\mathbf{x},\Theta_y+A\Theta_xA^{'}+\Sigma)=\mathcal{N}(A\mathbf{x},\Theta)$, so
	\begin{eqnarray}\label{p(xtp1|doxt)}
		p(\mathbf{y}|\mathbf{x})=\displaystyle\frac{1}{\det(\Theta)^\frac{1}{2}(2\pi)^\frac{m}{2}}\exp\left\{-\frac{1}{2}(\mathbf{y}-A\mathbf{x})^{'}\Theta^{-1}(\mathbf{y}-A\mathbf{x})\right\}
	\end{eqnarray}
   in which $\Theta=\Theta_y+A\Theta_xA^{'}+\Sigma$.
   It's obvious that
	\begin{eqnarray}\label{DegeneracyDeterminismEI}
		EI=\ln\left(\frac{|{\rm pdet}(A^{'}(\Theta_y+A\Theta_xA^{'}+\Sigma)^{-1}A)|^\frac{1}{2}L^n}{(2\pi)^\frac{n}{2}e^\frac{m}{2}}\right).
	\end{eqnarray}
\end{proof}

\textbf{Proof of Theorem \ref{thm.Causal-emergence}}
\begin{proof}
	For the micro dynamical systems like Equation (\ref{MicroDynamics}) and macro dynamical systems like Equation (\ref{MacroDynamics}) after coarse-graining, we can calculate the micro effective information $\mathcal{J}_m$ and macro effective information $\mathcal{J}_M$ of two kinds of dynamical systems separately. In this way, causal emergence
	\begin{eqnarray}\label{dCE}
		\begin{aligned}
			\Delta\mathcal{J}&=\mathcal{J}_M-\mathcal{J}_m=\frac{EI_M}{k}-\frac{EI_m}{n}\\
			&=\frac{1}{k}\ln\displaystyle\frac{|\det(A_M)|L^k}{\displaystyle \det(\Sigma_M)^\frac{1}{2}(2\pi e)^\frac{k}{2}}-\frac{1}{n}\ln\displaystyle\frac{|\det(A)|L^n}{\displaystyle \det(\Sigma)^\frac{1}{2}(2\pi e)^\frac{n}{2}}\\
			&=\frac{1}{k}\mathop{\left(\ln\displaystyle\frac{1}{\displaystyle \det(\Sigma_M)^\frac{1}{2}(2\pi e)^\frac{k}{2}}+\ln\displaystyle\left(|\det(A_M)|L^k\right)\right)}_{Determinism Information \quad Degeneracy Information }\\&-\frac{1}{n}\left(\ln\displaystyle\frac{1}{\displaystyle \det(\Sigma)^\frac{1}{2}(2\pi e)^\frac{n}{2}}+\ln\displaystyle\left(|\det(A)|L^n\right)\right)\\
			&=\mathop{\ln\frac{|\det(WAW^\dagger)|^\frac{1}{k}}{|\det(A)|^\frac{1}{n}}}_{Degeneracy Emergence}+\mathop{\ln\frac{|\det(\Sigma)|^\frac{1}{2n}}{|\det(W\Sigma W^{T})|^\frac{1}{2k}}}_{Determinism Emergence}.
		\end{aligned}
	\end{eqnarray}
\end{proof}

\section[\appendixname~\thesection]{Optimization of Macrodynamic Parameter Matrix}\label{appendixAM}
To maximize causal emergence, we need to use relevant content from matrix theory to optimize coarse-graining strategies.\\

\textbf{Proof of Lemma \ref{thm.Causal-emergence}}
\begin{proof}
	Perform eigendecomposition on matrix $A$, $A=V\Lambda V^{-1}$, $\Lambda={\rm diag}(\lambda_1, \dots, \lambda_n)\in \mathcal{R}^{n\times n}$ is the diagonal matrix of eigenvalues, $V\in \mathcal{R}^{n\times n}$ is the eigenvector matrix corresponding to the eigenvalue matrix. Determinant
	\begin{eqnarray}\label{determinant1}
		\det(WAW^\dagger)=\det(WV\Lambda V^{-1}W^\dagger)=\det(\tilde{W}\Lambda \tilde{W}^\dagger)=\displaystyle\sum_m\left( \gamma_m\prod_{\lambda_i\in L_m}\lambda_i\right),
	\end{eqnarray}
	in which $L_m=\{\lambda_{N_{m,1}}, \lambda_{N_{m,2}}, \dots, \lambda_{N_{m,k}}\}\subseteq\{\lambda_1,\dots,\lambda_n\}$ are new sets formed by arbitrarily choosing $k$ eigenvalues from $n$ eigenvalues of $A$. The sorting of these sets is based on the sum of the elements in the set, in descending order. $N_m\subseteq\{1,\dots,n\}$, the number of elements in sets $|N_m|,|L_m|=k$, $m=1,2,\dots,C_n^k$. $\gamma_m$ are coefficients determined by $W$, $\gamma_m>0$, $\sum_{m}\gamma_m=1$. It is not difficult to find that when $\gamma_m=1$, matrix $A$ only retains the eigenvalues in $L_m$ after transformation. 
	
	The calculation of $\gamma_1,\dots,\gamma_{C_n^k}$ are as follows. Set $WV=\tilde{W}=(w_{ij})_{k\times n}$, then $W^\dagger V^{-1}=\tilde{W}^\dagger=(v_{ij})_{n\times k}$. According to the principle of determinant calculation, we can divide $\gamma_m$ into two parts
	\begin{eqnarray}\label{coefficient1}
		\gamma_m=\gamma_{m,1}-\gamma_{m,2},
	\end{eqnarray}
	$\gamma_{m,1},\gamma_{m,2}$ correspond to the sum of terms with coefficients of $1$ and $-1$ in the determinant calculation, respectively. So
	\begin{eqnarray}\label{coefficient2}
		\gamma_{m,1}=\displaystyle\sum_{l=0}^{k-1}\left[\sum_{s\in\{s\}_m}\left(\prod_{i=1}^{k}w_{i,s_i}v_{s_i,j}\right)\right], j=
		\begin{cases}
			i+l, i+l&\leq k,\\
			i+l-k, i+l&>k.
		\end{cases}\\
		\gamma_{m,2}=\displaystyle\sum_{l=2}^{k+1}\left[\sum_{s\in\{s\}_m}\left(\prod_{i=1}^{k}w_{i,s_i}v_{s_i,j}\right)\right], j=
		\begin{cases}
			l-i, l-i&> 0,\\
			l-i+k, l-i&\geq 0.
		\end{cases}
	\end{eqnarray}
	$s=\{s_1,s_2,...s_k\}$ is a new sequence formed by randomly arranging all elements in $N_m$, $\{s\}_m$ is a set composed of such elements $s$. Number of elements in the set $|\{s\}_m|=k!$. 

    If all eigenvalues of $A$ satisfy $\lambda_1\geq\lambda_2\geq\dots\geq\lambda_n$, we can simultaneously find the maximum and minimum values of $\det(WAW^\dagger)$. When $\gamma_1=1$, only the maximum $k$ eigenvalues are retained, and $\det(WAW^\dagger)$ takes the maximum value $\prod_{i=1}^{k}\lambda_i$. On the contrary, when $\gamma_m=1$, only the smallest $k$ eigenvalues are retained, and the $\det(WAW^\dagger)$ takes the minimum value $\prod_{i=n-k+1}^{n}\lambda_i$. So 
\begin{eqnarray}\label{Maximizing-determinant}
    \displaystyle\prod_{i=n-k+1}^{n}\lambda_i\leq\det(WAW^\dagger)\leq\prod_{i=1}^{k}\lambda_i.
\end{eqnarray}
	
	From the proof of \ref{Maximizing-determinant}, we know that $\gamma_m$ are coefficients determined by $W$ and when $\gamma_m=1$, matrix $A$ only retains the eigenvalues in $L_m$. So when the eigenvalues of matrix $A$ are not specified to be all positive real numbers, the modulus length satisfies $|\lambda_1|\geq|\lambda_2|\geq\dots\geq|\lambda_n|$, then
\begin{eqnarray}\label{determinantabs}
 \begin{aligned}
    |\det(WAW^\dagger)|=\left|\displaystyle\sum_m\left( \gamma_m\prod_{\lambda_i\in L_m}\lambda_i\right)\right|\leq\displaystyle\sum_m \gamma_m\left|\prod_{\lambda_i\in L_m}\lambda_i\right|=\displaystyle\sum_m \gamma_m\left(\prod_{\lambda_i\in L_m}\left|\lambda_i\right|\right)\leq\displaystyle\prod_{i=1}^k|\lambda_i|,
 \end{aligned}
\end{eqnarray}
	We directly study the situation of eigenvalues in the complex field. When $L_m={\lambda_1,\dots,\lambda_k}$, the determinant $|\det(WAW^\dagger)|$ reach its maximum value. The decrease of $\gamma_m$, the intervention of smaller absolute values of eigenvalues $|\lambda_i|, i\geq k$ leads to the decrease of $|\det(WAW^\dagger)|$. 
\end{proof}

\textbf{Proof of Theorem \ref{JASigmaV}}
\begin{proof}
	According to (\ref{determinant1}) and Jensen's inequality, we can obtain inequalities
	\begin{eqnarray}\label{CEG1}
		\begin{aligned}
			|\det(W\Sigma W^\dagger)|^{-\frac{1}{2}}=\left[\sum \left(\gamma_m \prod\kappa_i\right)\right]^{-\frac{1}{2}}\leq\left[\sum \gamma_m\left( \prod\kappa_i\right)^{-\frac{1}{2}}\right]=|\det(W\Sigma^{-\frac{1}{2}} W^\dagger)|,
		\end{aligned}
	\end{eqnarray}
    if and only if the eigenvalues are equal, the equal sign holds. By utilizing this inequality and combining it with the relevant properties of linear algebra, we can conclude that
	\begin{eqnarray}\label{CEG}
	\begin{aligned}
		\Delta\mathcal{J}&=\ln\frac{|\det(WAW^\dagger)|^\frac{1}{k}}{|\det(W\Sigma W^\dagger)|^\frac{1}{2k}}-\ln\frac{|\det(A)|^\frac{1}{n}}{|\det(\Sigma)|^\frac{1}{2n}}\\
		&= \frac{1}{k}\ln(|\det(W\Sigma W^\dagger)|^{-\frac{1}{2}}|\det(WAW^\dagger)|)-\frac{1}{n}\ln|\Sigma^{-\frac{1}{2}}A|\\
		&\leq \frac{1}{k}\ln(|\det(W\Sigma^{-\frac{1}{2}} W^\dagger)||\det(WAW^\dagger)|)-\frac{1}{n}\ln|\Sigma^{-\frac{1}{2}}A|\\
		&= \frac{1}{k}\ln|\det(W\Sigma^{-\frac{1}{2}} W^\dagger WAW^\dagger)|-\frac{1}{n}\ln|\Sigma^{-\frac{1}{2}}A|\\
		_{(\Sigma=VKV^T, A=V\Lambda V^T)}&\leq \frac{1}{k}\ln|\det(W\Sigma^{-\frac{1}{2}} AW^\dagger)|-\frac{1}{n}\ln|\Sigma^{-\frac{1}{2}}A|\\
		_{(W=(\tilde{W}_k,O_{k\times{(n-k)}})V^{T})}&\leq\frac{1}{k}\sum_{i=1}^{k}\ln\displaystyle|\tilde{d}_i|-\frac{1}{n}\sum_{i=1}^{n}\ln\displaystyle|\tilde{d}_i|
	\end{aligned}
\end{eqnarray}
When $A$ and $\Sigma$ share the same $n$ eigenvectors, that is to say, $\Sigma=VKV^T$ and $A=V\Lambda V^T$, where $V$ is the matrix formed by the shared eigenvectors, and $K$ and $\Lambda$ are the eigenvalues for $A$ and $\Sigma$, respectively. The eigenvalues of $\Sigma^{-1/2}A$ are equal to the product of the eigenvalues of $A$ and $\Sigma^{-1/2}$. $W=(\tilde{W}_k,O_{k\times{(n-k)}})V^{T}$ represents that both $A$ and $\Sigma^{-1/2}$ retain the eigenvalues with the highest norm, the second and third unequal signs can also be taken as equal signs.
\end{proof}

\section[\appendixname~\thesection]{Loss function and causal emergence}\label{LossandJ}
In this section, we first prove that $S_D$ is the supremum of $L_D$.\\

\textbf{Proof of Lemma \ref{thm.supdynamical-loss}}
\begin{proof}
	 Based on the theory of matrix modules, when $||x_t||^*=\mathop{\sup}_{t}||x_t||<\infty$,
	\begin{eqnarray}\label{LSD}
		\begin{aligned}
		L_D&=||Ax_t+\varepsilon_t-(W^\dagger WAW^\dagger Wx_t+W^\dagger W\varepsilon_t)||\\
		&=||(A-W^\dagger WAW^\dagger W)x_t+(I_n-W^\dagger W)\varepsilon_t||\\
		&\leq||(A-W^\dagger WAW^\dagger W)x_t||+||(I_n-W^\dagger W)\varepsilon_t||\\
		&\leq||(A-W^\dagger WAW^\dagger W)||_F||x_t||+||(I_n-W^\dagger W)||_F||\varepsilon_t||\\
		&=||(A-W^\dagger WAW^\dagger W)||_F||x_t||+(n-k)||\varepsilon_t||\\
  &\leq||(A-W^\dagger WAW^\dagger W)||_F||x_t||^*+(n-k)||\varepsilon_t||
  =S_D.
		\end{aligned}
	\end{eqnarray}
\end{proof}
Next, we prove that $W$ at the minimum value of the supremum $S_D$ can maximize the degeneracy emergence.\\

\textbf{Proof of Lemma \ref{thm.CESD}}
\begin{proof}
	$|\lambda_1|\geq|\lambda_2|\geq\dots\geq|\lambda_n|\geq 0$ are $n$ modulus of eigenvalues of matrix $A$ sorted from largest to smallest. $\Lambda={\rm diag}(\lambda_1, \dots, \lambda_n)\in \mathcal{R}^{n\times n}$. $\tilde{W}=WV$, $\tilde{W}^\dagger \tilde{W}$ and $I-\tilde{W}^\dagger \tilde{W}$ are both symmetric idempotent matrices. $rk(\tilde{W}^\dagger \tilde{W})=k$, $I-\tilde{W}^\dagger \tilde{W}=(r_{ij})_{n\times n}$ has the following properties:
	\begin{enumerate}[(1)]
		\item $\displaystyle\sum_{l=1}^{n}r_{il}r_{lj}=r_{ij}$.
		\item $||I-\tilde{W}^\dagger \tilde{W}||_F=\displaystyle\sqrt{\sum_{i,j=1}^{n,n}r_{ij}^2}=k$.
		\item $\forall l, 0<l\leq n, \displaystyle\sqrt{\sum_{i,j=1}^{n,n}(r_{il}r_{lj})^2}=1$
	\end{enumerate}
	Based on the theory of matrix modules, 
	\begin{eqnarray}\label{CESD}
		\begin{aligned}
			S_D&=||(A-W^\dagger WAW^\dagger W)||_F||x_t||+(n-k)||\varepsilon_t||\\
			_{(\tilde{W}=WV)}&=||V(\Lambda-\tilde{W}^\dagger \tilde{W}\Lambda \tilde{W}^\dagger \tilde{W})V^{-1}||_F||x_t||+(n-k)||\varepsilon_t||\\
			&=||\Lambda-\tilde{W}^\dagger \tilde{W}\Lambda \tilde{W}^\dagger \tilde{W}||_F||x_t||+(n-k)||\varepsilon_t||\\
			&=||(I_n-\tilde{W}^\dagger \tilde{W})\Lambda(I_n-\tilde{W}^\dagger \tilde{W})+2(I_n-\tilde{W}^\dagger \tilde{W})\Lambda\tilde{W}^\dagger \tilde{W}||_F||x_t||+(n-k)||\varepsilon_t||\\
			&\geq ||(I_n-\tilde{W}^\dagger \tilde{W})\Lambda(I_n-\tilde{W}^\dagger \tilde{W})||_F||x_t||+(n-k)||\varepsilon_t||\\
			&=\displaystyle\sqrt{\sum_{0<i,j\leq n}\left|\sum_{l=1}^n r_{il}r_{lj}\lambda_l\right|^2}||x_t||+(n-k)||\varepsilon_t||\\
			&\geq\displaystyle\sqrt{\sum_{0<i,j\leq n}\sum_{l=1}^n \left|r_{il}r_{lj}\lambda_l\right|^2}||x_t||+(n-k)||\varepsilon_t||\\
			&\geq\displaystyle\sqrt{\sum_{0<i,j\leq n}\sum_{l=k+1}^n \left|r_{il}r_{lj}\lambda_l\right|^2}||x_t||+(n-k)||\varepsilon_t||\\
			&\geq\displaystyle\sqrt{\sum_{l=k+1}^n \left|\lambda_l\right|^2}||x_t||+(n-k)||\varepsilon_t||.
		\end{aligned}
	\end{eqnarray}
 If and only if the maximum $k$ eigenvalues are preserved in $\hat{A}$, only the $n-k$ eigenvalues with the smallest modulus are retained $S_D$, the equal sign holds. $W=(\tilde{W}_k,O_{k\times{(n-k)}})V^{-1}$, exactly the same solution set as the degeneracy emergence at its maximum $\Delta\mathcal{J}_1^*$. 
\end{proof}


\bibliographystyle{ieeetr}
\bibliography{jiaoda11}

\end{document}